\documentclass[11pt]{article}
\usepackage{fullpage}
\usepackage{amsmath,amsthm,amssymb}
\usepackage{xcolor}
\usepackage{algpseudocode,algorithm,algorithmicx}
\usepackage{mathtools}
\usepackage{sectsty}
\usepackage{hyperref}
\usepackage[nameinlink,capitalise]{cleveref}
\usepackage{thm-restate}
\usepackage{natbib}
\usepackage{bm}
\usepackage{bbm}

\makeatletter
\def\@fnsymbol#1{\ensuremath{\ifcase#1\or \dagger\or \ddagger\or
   \mathsection\or \mathparagraph\or \|\or **\or \dagger\dagger
   \or \ddagger\ddagger \else\@ctrerr\fi}}
\makeatother
\newcommand*\samethanks[1][\value{footnote}]{\footnotemark[#1]}

\definecolor{myblue}{rgb}{0.15, 0.1, 0.95}
\definecolor{mygreen}{rgb}{0.15, 0.65, 0.25}
\definecolor{myred}{rgb}{0.75, 0.25, 0.15}
\hypersetup{
    colorlinks=true,
    linkcolor=myred,
    citecolor = myblue,      
}

\newcommand*\Let[2]{\State #1 $\gets$ #2}
\algrenewcommand\algorithmicrequire{\textbf{Precondition:}}

\newtheorem{theorem}{Theorem}[section]
\newtheorem{lemma}[theorem]{Lemma}
\newtheorem{corollary}[theorem]{Corollary}

\theoremstyle{definition}
\newtheorem{definition}[theorem]{Definition}
\newtheorem{example}[theorem]{Example}

\providecommand{\Prob}[2]{\ensuremath{{\Pr_{#1}\!\left[#2\right]}}}
\providecommand{\Ex}[2]{\ensuremath{{\mathbb{E}_{#1}\!\left[#2\right]}}}

\renewcommand{\bar}[1]{\mkern 2.5mu\overline{\mkern-2.5mu#1\mkern-2.5mu}\mkern 2.5mu}

\newcommand{\mathsc}[1]{{\normalfont\textsc{#1}}}
\newcommand{\lin}{\mathsc{Lindahl}}
\newcommand{\nw}{\mathsc{NW}}
\newcommand{\idep}{\mathsc{Round}}
\newcommand{\ndep}{\mathsc{NWRound}}

\newcommand{\km}[1]{\textcolor{teal}{#1}}

\newcommand{\g}{\gamma}
\newcommand{\veps}{\varepsilon/b}
\newcommand{\eps}{\varepsilon}
\newcommand{\cost}{\mathrm{Cost}}
\newcommand{\ph}{\hat{\partial}}
\newcommand{\p}{\partial}
\newcommand{\ux}{U_i(\mathbf{x})}
\newcommand{\mx}{\mathbf{x}}
\newcommand{\poly}{\mbox{poly}}
\newcommand{\destat}{\frac{\eps^6}{64n\cdot m^5}}
\newcommand{\dstate}{\frac{\eps^7 b}{312 m^6}}

\title{Approximate Core for Committee Selection via Multilinear Extension and Market Clearing}
\author{Kamesh Munagala\thanks{Department of Computer Science, Duke University, Durham, NC 27708-0129. Emails:
\texttt{kamesh@cs.duke.edu, yiheng.shen@duke.edu, knwang@cs.duke.edu, zhiyi.wang@duke.edu}.
This work is supported by NSF grants CCF-1637397 and CCF-2113798, ONR award N00014-19-1-2268, and DARPA award FA8650-18-C-7880.} \and Yiheng Shen\samethanks[1] \and Kangning Wang\samethanks[1] \and Zhiyi Wang\samethanks[1]}
\date{}

\begin{document}
\maketitle

\begin{abstract}
Motivated by civic problems such as participatory budgeting and multiwinner elections, we consider the problem of \emph{public good} allocation: Given a set of indivisible projects (or candidates) of different sizes, and voters with different monotone utility functions over subsets of these candidates, the goal is to choose a budget-constrained subset of these candidates (or a committee) that provides fair utility to the voters. The notion of fairness we adopt is that of \emph{core stability} from cooperative game theory: No subset of voters should be able to choose another \emph{blocking committee} of proportionally smaller size that provides strictly larger utility to all voters that deviate.  The core provides a strong notion of fairness, subsuming other notions that have been widely studied in computational social choice.

It is well-known that an exact core need not exist even when utility functions of the voters are additive across candidates. We therefore relax the problem to allow \emph{approximation}: Voters can only deviate to the blocking committee if after they choose any extra candidate (called an \emph{additament}), their utility still increases by an $\alpha$ factor. If no blocking committee exists under this definition, we call this an $\alpha$-core.

Our main result is that an $\alpha$-core, for $\alpha < 67.37$, always exists when utilities of the voters are arbitrary monotone submodular functions, and this can be computed in polynomial time. This result improves to $\alpha < 9.27$ for additive utilities, albeit without the polynomial time guarantee. 
Our results are a significant improvement over prior work that only shows logarithmic approximations for the case of additive utilities. We complement our results with a lower bound of $\alpha > 1.015$ for submodular utilities, and a lower bound of any function in the number of voters and candidates for general monotone utilities.

\end{abstract}

\section{Introduction}
Consider the following scenario: A city is deciding what public projects to fund using its limited budget $b$. There is a list of $m$ candidate projects (forming a set $C$), where each $j \in C$ is associated with a cost $s_j$. The city needs to select a subset $O$ of these projects whose total cost is at most the budget, that is, $\sum_{j \in O} s_j \le b$. There are $n$ residents or {\em voters} (forming a set $V$) in the city, and each of them has preferences on how the city should spend its budget. These preferences may not be perfectly compatible with each other: For instance, families with children may prefer a public school, while others may prefer a park; people living in the east may prefer projects there, and so for those living in the west. It is desirable to have a fair process to decide on the projects to pay for. In participatory budgeting, voters express their preferences through their votes and influence the decision process. The paradigm has been implemented in numerous cities across the world~\citep{cabannes2004participatory,aziz2021participatory,pbstanford,knapsackVoting}.


A similar problem is seen in multiwinner elections~\citep{AzizChapter,EndrissBook,VinceBook,CC,thiele1895om}, where voters select a committee of size $b$ from $m$ candidates. Each voter holds her own opinions on the committees, and a fair method to incorporate the preferences of all voters is called for. Mathematically, it is a special case of the participatory budgeting problem, where each candidate has the same ``cost''. 

In these settings, simple methods to aggregate the preferences may have drawbacks. In majority voting, the utilities of a coherent minority group might be entirely ignored in favor of the majority. Utilitarian ways (maximizing the sum of utilities) may overly focus on a certain group, disregarding the welfare of the vast majority. Indeed, many practically implemented voting rules for choosing parliaments and civic body members, such as the well-known Single Transferable Voting (STV)~\citep{Tideman}, attempt to address precisely this issue. We naturally ask: 
\begin{quote}
What is a fair solution for participatory budgeting and multiwinner elections, and how do we reach such a solution?
\end{quote}

\subsection{The Core and Its Multiplicative Approximation}
Recall that there are $n$ voters forming a set $V$, and $m$ candidates forming a set $C$, where candidate $j$ has size (or cost) $s_j$. We need to choose a subset $O$ of candidates with total size at most $b$ (that is, $\sum_{j \in O} s_j \le b$). Following social choice terminology, this subset of candidates is called a {\em committee}, and the problem of choosing a budget-constrained committee is called the {\em committee selection} problem. Denote the utility of voter $i$ for a committee $T$ by a {\em utility function} $u_i(T)$. We will assume this function is non-negative and monotone, with $u_i(\varnothing) = 0$.

Although there are copious notions of fairness for committee selection, the \emph{core} is a classic and influential one among them. This idea has existed for more than a century~\citep{Droop,thiele1895om,lindahl1958just}, and serves as a strong notion of proportional representation. Towards defining this concept, imagine we split the size $b$ among the voters, so that each voter has an {\em endowment} of $\frac{b}{n}$ that they can use to ``buy'' candidates. A candidate of size $s_j$ requires an endowment of $s_j$ to ``buy''. A committee $O \subseteq C$ with total size at most $b$ is said to be in the core, if no subset $S$ of voters can deviate and purchase another committee $T \subseteq C$ by pooling their endowments, so that each voter in $S$ prefers the new committee $T$ to the original one $O$. Note that the total endowment of $S$ is $|S| \frac{b}{n}$, so that this set of voters can buy a committee $T$ of size at most $|S| \frac{b}{n}$.  Formally,

\begin{restatable}[Core]{definition}{core}
\label{def:core}
A committee $O$ is in the \emph{core} if there is no $S \subseteq V$ and committee $T \subseteq C$ with $\sum_{j \in T} s_j \leq \frac{|S|}{n} \cdot b$, such that $u_i(T) > u_i(O)$ for every $i \in S$.
\end{restatable}

The core has a ``fair taxation'' interpretation~\citep{lindahl1958just,foley1970lindahl}. The quantity $\frac{b}{n}$ can be thought of as the tax contribution of a voter, and a committee in the core has the property that no sub-group of voters could have spent their share of tax money in a way that \emph{all} of them were better off. As such it subsumes notions of fairness such as Pareto-optimality, proportionality, and various forms of justified representation~\citep{JR,Sanchez,PJR2018} that have been extensively studied in multiwinner election and fairness literature.

Despite the satisfying properties of the core, its strength is also its limitation: Even in the simple setting of unit sizes, integer budget, and additive utilities (the so-called {\em approval-set} setting with general utilities), the core can be empty.  (See for example, \citep{DBLP:conf/ec/FainMS18}.)

A natural approach to circumvent this problem is to show the existence of a committee that \emph{multiplicatively} approximates the core. We define the $\alpha$-core as follows.




\begin{restatable}[$\alpha$-Core]{definition}{utilitycore}
\label{def:alphaCore}
A committee $O$ is in the \emph{$\alpha$-core} if there is no $S \subseteq V$ and $T \subseteq C$ with $\sum_{j \in T} s_j \leq \frac{|S|}{n} \cdot b$, such that $u_i(T) > \alpha \cdot u_i(O \cup \{q\})$ for every $i \in S$ and $q \in C$. We call the candidate $q$ an \emph{additament}.
\end{restatable}

Note that we introduce the additament in the definition, since no multiplicative approximation is possible without it even in the setting with unit candidate sizes and additive utilities. This follows from examples in previous work~\citep{DBLP:conf/ec/FainMS18,DBLP:journals/teac/ChengJMW20}.
The idea of a bicriteria (multiplicative and additive) approximation to utilities was first presented by \citet{DBLP:conf/ec/FainMS18}.  \citet{DBLP:journals/corr/PetersPS20} present an almost identical definition as \cref{def:alphaCore}, except that the additament $q$ must come from the set $T$. This makes their definition more restrictive and the $\alpha$-core smaller. They show that when utilities are additive, an $O\left(\log \frac{u_{\mathrm{max}}}{u_{\mathrm{min}}}\right)$-core solution not only exists, but can also be computed in polynomial time, where $u_{\max}$ and $u_{\min}$ are the largest and smallest non-zero utilities any voter has for any feasible committee.

\subsection{Our Results}
Our main result in \cref{sec:poly} is to show the existence of an $O(1)$-core when the utility functions of the voters are monotone and submodular. Submodular utilities capture the notion of ``diminishing returns'', and are therefore well-motivated in participatory budgeting and multiwinner election settings where each additional project or candidate provides diminishing marginal utility.

Our result is a significant improvement over the previous work mentioned above~\citep{DBLP:journals/corr/PetersPS20}, which only presents a logarithmic approximation for the restricted case of additive utilities. Our main result is formally the following. 

\begin{restatable}[Main Theorem]{theorem}{main}
\label{thm:main}
For monotone submodular utilities, a $67.37$-core is always non-empty. One such solution can be computed in polynomial time.
\end{restatable}


Our second result in \cref{sec:additive} is the following improved result for the well-studied special case of additive utilities:

\begin{restatable}[Additive Utilities]{theorem}{additive}
\label{cor:additive}
For additive utilities, a $9.27$-core is always non-empty.
\end{restatable}

Unlike \cref{thm:main}, we do not know how to implement the algorithm in \cref{cor:additive} in polynomial time. We remark that the previous two results can be combined to show that a $15.2$-core solution for additive utilities can be computed in polynomial time, without providing the details.



\paragraph{Lower Bounds.} The next natural question is whether our results can be extended to arbitrary monotone utilities. We show that this is not possible. We present an example in \cref{sec:general} to show that an $O(1)$-core (or even any $f(n, m)$-core where $n$ is the number of voters and $m$ is the number of candidates) may not exist for general monotone utilities. 


\begin{restatable}[General Utility Lower Bound]{theorem}{general}
\label{thm:general_lb}
For general monotone utilities, for any function $\varphi : \mathbb{Z}^+ \times \mathbb{Z}^+ \to \mathbb{R}^+$, a $\varphi(n, m)$-core can be empty.
\end{restatable}

We also show that the case of submodular utilities does need a multiplicative approximation to the core, and the $1$-core can be empty. This justifies the form of \cref{thm:main} that involves both a multiplicative approximation and an additament.

\begin{restatable}[Submodular Lower Bound]{theorem}{submodularlb}
\label{thm:submodular_lb}
For monotone submodular utilities, a $1.015$-core can be empty.
\end{restatable}

\subsection{Technical Contributions} 
The main technical idea in our paper is to consider {\em fractional relaxations} of the problem, via considering outcomes where candidates are fractionally allocated. For submodular utilities, we extend the utility to fractional values via the so-called multilinear extension~\citep{Vondrak,CalinescuCPV11}. For additive utilities, the extension is straightforward. We then construct an exact or approximate core outcome for the fractional version. We finally round this solution to find an integer core. We now describe these steps in more detail.

\paragraph{Fractional Core.} We use two entirely different techniques for constructing fractional core outcomes to show \cref{thm:main} and \cref{cor:additive}. For both solutions, we extend the standard core property and show a stronger {\em approximation property} (\cref{thm:grad} and \cref{lem:St}) on how the size of a coalition of voters relates to the budget of the committee to which they can deviate to and obtain a multiplicative factor larger utility.

The first approach that we use for arbitrary submodular functions is via optimizing the classic Nash Welfare objective~\citep*{Nash}, that maximizes the product of utilities of the voters. Though the Nash Welfare is a convex optimization problem when utilities are continuous, monotone, and concave, in our case, the utility is now defined as the multilinear extension and this is only concave in positive directions~\citep{CalinescuCPV11}. Nevertheless, we show a simple continuous local search procedure for the Nash Welfare objective, such that the local optimum is a $2$-approximate fractional core. For our rounding, we need a stronger approximation property that we prove in \cref{thm:grad}. This property may be of independent interest. 

The second approach that we use for additive utilities, finds an {\em exact} fractional core. This uses a classic market clearing solution is termed the {\em Lindahl equilibrium}~\citep*{lindahl1958just,foley1970lindahl}. In this equilibrium, candidates can be chosen to any fraction (even greater than one). Each candidate is assigned a per-voter ``price'', and the voters are assigned an endowment of $\frac{b}{n}$ just as in \cref{def:core}. At these prices, if the voters choose their utility maximizing allocation subject to spending their endowment, then (1) they all choose the same fractional committee, and they spend their entire endowment; and (2) for each chosen candidate, the total price collected exactly pays for that candidate. \citet{foley1970lindahl} shows via a fixed point argument that such an equilibrium always exists when each voter's utility for fractional committees is any continuous, monotone, and concave function. Further, the resulting equilibrium is a core outcome.

\paragraph{Rounding.} We then construct the integer solution by rounding of the fractional core. Since the allocation is common to all voters, it is easier to work with rounding processes that are oblivious to the utility functions, and we therefore simply use randomized rounding. If we randomly round the allocation to integer values, the expected utility of any voter is preserved. This can now be combined with Chernoff-type bounds for multilinear extensions~\citep{ChekuriVZ10} to argue that a constant factor of voters see utility that is at least a constant factor of the corresponding utility in the fractional solution. For additive utilities, we use dependent rounding~\citep{Gandhi,DBLP:journals/talg/ByrkaPRST17} to provide an improved constant here. 

This brings up the next hurdle: The remaining unsatisfied set of voters can still deviate to a different coalition. Our second ingredient is the iterative rounding framework of~\citet{DBLP:conf/stoc/JiangMW20}: We recurse on the unsatisfied voters to find a fractional solution with a smaller budget on total size. The key idea is that every voter is satisfied at some level of the recursion. We now use the approximation property of the fractional solution at each level of recursion to argue that the number of voters deviating to any fixed committee that provides them a constant factor more utility is bounded. Summing this argument over all levels completes the proof. 

Our overall argument in both cases -- submodular and additive utilities -- is quite delicate, and crucially requires strong approximation properties of the fractional solutions of Lindahl equilibrium and locally optimal Nash Welfare. This forms our main technical contribution.



\subsection{Related Work}
\label{sec:related}
\paragraph{Proportionality and the Core.} One classic objective in committee selection is achieving fairness via proportionality, where different demographic slices of voters feel they have been represented fairly. This general idea dates back more than a century~\citep{Droop}, and has recently received significant attention~\citep{CC,Monroe,Brams2007,JR,Sanchez,PJR2018}. In fact, there are several elections, both at a group level and a national level, that attempt to find committees (or parliaments) that provide approximately proportional representation. For instance, the popular Single Transferable Vote (STV) rule is used in parliamentary elections in Ireland and Australia, and in several municipal elections in the USA. This rule attempts to find a proportional solution.  

A long line of recent literature has studied the complexity and axiomatization of voting rules that achieve proportionality; see~\citep{AzizChapter,EndrissBook,VinceBook} for recent surveys. Proportionality in committee selection arises in many other applications outside of social choice as well. For example, consider a shared cache for data items in a multi-tenant cloud system, where each data item is used by several long-running applications~\citep{ROBUS,Psomas}. Each data item can be treated as a candidate, and each application as a voter whose utility for an item corresponds to the speedup obtained by caching that item. In this context, a desirable caching policy provides proportional speedup to all applications.

The core represents the ultimate form of proportionality: Every demographic of voters feel that they have been fairly represented and do not have the incentive to deviate and choose their own committee of proportionally smaller size which gives all of them higher utility. In the typical setting where these demographic slices are not known upfront, the notion of core attempts to be fair to all subsets of voters. The work of~\citet{DBLP:conf/ec/MunagalaSW21} formally argues that in certain multiwinner election settings, the core also approximately optimizes simpler diversity measures of the resulting committee.

\paragraph{Fisher Markets.} Our fractional solutions are superficially related to the Fisher market equilibrium~\citep{Fisher,Nash,ArrowD} when divisible items need to be {\em allocated} to agents, and agents' utilities are additive. For the Fisher market, the optimum Nash Welfare solution  finds market clearing prices. However, in a Fisher market, the prices are common to the agents while the allocations are different, while in a Lindahl equilibrium, the prices are per-voter while the allocation (or committee) is common and provides shared utility to all the voters. This is a key difference -- the Fisher market has a polynomial time algorithm via convex programming~\citep{EG}, while no polynomial time algorithm is known for the Lindahl equilibrium even when candidates have unit sizes and voters' utilities are additive (or linear till the maximum size of the candidate). Similarly, though the Nash Welfare solution finds market clearing prices for the Fisher market via strong duality, in the case of public goods, there is no obvious way to interpret the dual of the Nash Welfare solution as market clearing prices. Moreover, for submodular utilities and multilinear extensions, the Nash Welfare objective is no longer a convex program, so that strong duality does not apply.

\paragraph{Approximate Core.} In this paper, we have focused on approximating the utility voters obtain on deviating (see \cref{def:alphaCore}). As mentioned before, this notion first appeared in~\citep{DBLP:conf/ec/FainMS18}, and the notion of a single additament in approximation is due to~\citet{DBLP:journals/corr/PetersPS20}. The latter work present a logarithmic approximation for the special case of additive utilities. \citet{PetersS20} show that the well-known Proportional Approval Voting method~\citep{thiele1895om} achieves a $2$-core for the special case where the utilities are additive and candidates are unit size, with each voter having utility either zero or one for each candidate. This algorithm can be viewed as a discrete version of Nash Welfare, and in essence, \cref{thm:grad} extends this result to the case of submodular utilities and general costs, showing that it yields a $2$-core for the {\em fractional} case of multilinear extension via a polynomial time local search algorithm. The work of~\citet{ChenFLM19} presents a constant approximation for the $K$-clustering problem, where the committee is a set of $K$ centers in a metric space, and the cost of a voter is the distance to the closest center. However, these ideas do not extend to the committee selection problem we consider in this paper.

\citet{DBLP:conf/stoc/JiangMW20} consider a different notion of approximation: Instead of approximating the utility, they approximate the {\em endowment} that a voter can use to buy the deviating committee. Building on the work of~\citet{DBLP:journals/teac/ChengJMW20}, they show a different fractional relaxation, to which a $2$-approximation always exists. They then iteratively round this fractional solution to an integer solution that is a $32$-approximation for {\em all} monotone utility functions. The problem of approximating utilities is very different; indeed, \cref{thm:general_lb} shows we cannot hope to have a similar constant approximation for all utility functions. Nevertheless, we use the idea of iterative rounding from that work, albeit with an entirely different fractional solution and analysis. In effect, we showcase the power of iterative rounding as a unifying framework for finding approximate core solutions, regardless of the notion of approximation.

\paragraph{Rounding Techniques.} The notion of multilinear extension and correlation gap has been widely used in stochastic optimization~\citep{Shipra10}, mechanism design~\citep{Yan11}, and rounding~\citep{Vondrak,CalinescuCPV11,ChekuriVZ14}. Typically, it has been used to develop computationally efficient approaches; on the other hand, we demonstrate an application to showing a purely existential result.  Similarly, rounding of market clearing solutions have been used to show approximately fair allocations of indivisible goods among agents~\citep{BarmanK,ColeG}. The structure of these problems (common prices but different allocations) is very different from ours (common allocations and different prices), and we need different techniques. Again, in contrast with the resource allocation literature, we need the rounding just to show an existence result as opposed to a computational one.

\section{Model and Preliminaries}
Recall that $C$ is a set of $m$ candidates and $V$ is a set of $n$ voters. Each candidate $j \in C$ is associated with a size $s_j > 0$. For each $i \in V$ and each $T \subseteq C$, we denote by $u_i(T)$  the utility of voter $i$ on committee $T$.  We assume $u_i$ is a monotone, submodular set function with $u_i(\varnothing) = 0$:
\begin{itemize}
\item Monotonicity: \ $u_i(T) \le u_i(T') \ \ \forall T \subseteq T'$.
\item Submodularity: \ $u_i(T \cup \{t\}) - u_i(T) \ge u_i(T' \cup \{t\}) - u_i(T') \ \ \forall T \subseteq T', t \notin T'$.
\end{itemize}


We have a constraint of $b$ on the total size of the committee. We call a committee $O \subseteq C$ feasible (with respect to $b$) if: $\sum_{j \in O} s_j \leq b$.

Our goal is to find a committee in the $\alpha$-core for $\alpha = O(1)$. Recall its definition:
\utilitycore*


\paragraph{Fractional Allocations and Core.} In the sequel, we will consider continuous extensions of the committee and the utility function. We define a {\em fractional committee} as a $m$-dimensional vector $\mathbf{x} \ge 0$. The quantity $x_j$ denotes the fraction to which candidate $j$ is allocated. 

We denote by $\cost(\mathbf{x}) := \sum_{j \in C} x_j\cdot s_j$ as the cost of the allocation $\mathbf{x}$.  Without loss of generality, we assume that $\sum_{j \in C} s_j > b$. 

Similarly, we consider continuous utilities $U_i(\mathbf{x})$ for the voters, whose construction will be specified later. Given this notation, an $\alpha$-approximate fractional core, for $\alpha \ge 1$ is defined as follows:

\begin{definition}[$\alpha$-Approximate Fractional Core]
\label{def:core2}
A fractional committee $\mathbf{x} \ge 0$ with $\cost(\mathbf{x}) \le b$ lies in the $\alpha$-approximate \emph{fractional core} if there is no $S \subseteq V$ and allocation $\mathbf{z} \ge 0$ with $\cost(\mathbf{z}) \le \frac{|S|}{|V|} \cdot b$ such that $U_i(\mathbf{z}) > \alpha U_i(\mathbf{x})$ for all $i \in S$.
\end{definition}

When $\alpha = 1$, we call the fractional committee $\mathbf{x}$ as simply lying in the {\em fractional core}.

\section{Constant Approximate Core for Submodular Utilities}
\label{sec:poly}
Our overall algorithm proceeds via constructing approximate fractional core for a well-known continuous relaxation, called the {\em multilinear relaxation}. We use a continuous time local search procedure for the Nash Welfare objective to find this fractional solution. We show that it is almost in a $2$-approximate fractional core, and subsequently round it iteratively to find a solution in the approximate integer core. Our overall algorithm runs in time polynomial in the number of voters and candidates.  We first present the entire algorithm as an existence proof, and then delve into the details of the running time in \cref{sec:polyproof}.


\subsection{First Observations}
The following steps are without loss of generality, aiming to simplify future derivations in this section. We will assume $\eps \in (0,1/20)$ is a small constant that we choose later. In a proof of existence we can set $\eps = 0$; however, we need $\eps > 0$ to achieve polynomial running time. 
\begin{itemize}
\item 
Let $u_i^{\max} := \max_j\{u_i(\{j\})\}$, the maximum utility of voter $i$ with a single item. We scale down every $u_i(T)$ by a factor of $u_i^{\max}$ and denote this normalized utility as $u_i'(T) := \frac{u_i(T)}{u_i^{\max}}$. (We ignore every voter with $u_i(C) = 0$, since they will not deviate in any situation.) 
By submodularity of the utility functions, $u_i'(T)\in [0, m]$ for every $i$ and $T\subseteq C$. This step preserves the original problem, since every voter multiplicatively compares utilities of different committees in \cref{def:alphaCore} and different voters compare separately. Therefore, without loss of generality, we assume $\max_j\{u_i(\{j\})\} = 1$ and thus $u_i(T)\in [0, m], \forall T\subseteq C$.
\item We further assume $s_j\le b$  for each $j \in C$, as otherwise, the candidate cannot be selected in any feasible committee. 
\item Now define 
\[
C_{s} := \{j \mid s_{j}\le \eps b/m\}
\]
as a set of ``small enough'' candidates. We place all these candidates in the solution. Clearly all candidates in $C_s$ take up a total budget of at most $\eps b$ -- the remaining budget is at least $(1-\eps)b$. Our algorithm then works with candidates in:
\[
C_{\ell} := C \setminus C_{s} = \{j \mid s_{j} > \eps b/m\}.
\]
Note that each candidate $j$ in $C_{\ell}$ now has size $s_j > \eps b/m$. 
\end{itemize}

\subsection{The $2$-Approximate Fractional Core}
\label{sec:nash}
We now present a procedure that finds an approximate fractional core solution for an arbitrary budget $B \in \left[\frac{\eps}{5 m} b, b \right]$, and an arbitrary subset $W \subseteq V$ of voters. We present the algorithm and its optimality properties below; however, we defer the polynomial running time analysis to \cref{sec:polyproof}.

\subsubsection{Multilinear Extension}
Recall that the utility function $u_i(S)$ of voter $i$ for $S \subseteq C$ is non-negative, monotone, and submodular. We first extend these utilities to fractional allocations. A natural way to do that is to use the multilinear extension.

\begin{definition}[Multilinear Extension, \citep{CalinescuCPV11,Vondrak}]
	\label{def:concext}
	Given a monotone, non-negative submodular function $f$, its \emph{multilinear extension} $F$ is defined for any $\mathbf{x} \in [0,1]^m$ as:
	$$ F(\mathbf{x}) = \sum_{T \subseteq C} f(T) \left(\prod_{j \in T} x_j \right) \left(\prod_{j \notin T} (1-x_j) \right).$$
\end{definition}

We will apply the following property of this function.

\begin{lemma}[Concavity along Positive Directions, \citep{CalinescuCPV11}]
	\label{lem:conm}
	Given a monotone, non-negative submodular function $f$, its multilinear extension $F$ is concave along positive directions, {\em i.e.}, for all $\mathbf{x_0}$ and all $\mathbf{d}\ge \mathbf{0}$, we have that $F_{\mathbf{x_0},\mathbf{d}}(\lambda)=F(\mathbf{x_0}+\lambda\mathbf{d})$ is a monotone, non-negative and concave function of $\lambda$.
\end{lemma}

We denote the multilinear extension of voter $i$'s utility function $u_i$ for a fractional committee $\mathbf{x}$ as $U_i(\mathbf{x})$. The above lemma implies that $\frac{\partial U_i}{\partial x_j} \ge 0$ for all $i \in W, j \in C$. The next lemma upper bounds the gradient. Here, $\mathbf{x_{-j}}$ is $\mathbf{x}$ with the $j^\text{th}$ dimension removed, and $T \sim \mathbf{x_{-j}}$ means that $T$ is chosen by including $\ell \in C \setminus \{j\}$ independently with probability $x_{\ell}$. 

\begin{lemma}[Bounded Gradient, \citep{CalinescuCPV11}]
\label{lem:gradsize}
$$ \frac{\partial \ux}{\partial x_j} = \sum_{T \subseteq C \setminus \{j\}} \Pr_{T \sim \mathbf{x_{-j}}}\left[T\right] \cdot \left( u_i(T \cup \{j\}) - u_i(T) \right) \le \max_j u_i(\{j\}) = 1.$$
\end{lemma}

\subsubsection{Nash Welfare and Local Search}
The input of the algorithm consists of three parts: The utility function $u_i(T)$ for each committee $T\subseteq C$ and voter $i \in W$, the sizes of the candidates $\{s_j\}$ and the total budget $B \in \left[\frac{\eps}{5m} b, b \right]$ on candidates in $C_{\ell}$. Recall that we assume all candidates in $C_s$ have already been chosen. We assume w.l.o.g. that $\sum_{j \in C_{\ell}} s_j > B$.  In our algorithm, we treat the utility function as an oracle which returns the utility of an agent $i$ on a given committee in $O(1)$ time.

In our algorithm, we will lower bound the allocation $x_j$ by an amount $\underline{x}_j$ that is defined as follows.
\[
\underline{x}_j=
\begin{cases}
1 & j\in C_s,\\
\frac{B \cdot \eps}{\sum_{j\in C_{\ell}} s_j} & j\in C_{\ell}.
\end{cases}
\]

\paragraph{Nash Welfare.} We will find a local optimum to the following {\em Nash Welfare} program to find the fractional solution. 

\[ \mbox{Maximize} \ \sum_{i \in W} \log U_i(\mathbf{x}) \]
\[ \begin{array}{rcll}
\sum_{j \in C_{\ell}} s_j x_j & \le & B \\
x_j & \in & [\underline{x}_j,1] & \forall j \in C
\end{array} \]

Let $\phi_i(\mathbf{x}) = \log U_i(\mathbf{x})$, and $\phi(\mathbf{x}) = \sum_{i \in W} \phi_i(\mx)$. This is the objective value of the above program.

\paragraph{Local Search.}  We call the local search procedure as \nw{}$(C, W, \{U_{i}\}, B)$. This procedure has a {\em step-size} parameter $\delta = \dstate$.  We start with any allocation $\mathbf{x_0} \in [\underline{x}_j,1]^m$ with $\sum_{j \in C_{\ell}} s_j x_j = B$. Note that such an allocation always exists since we assumed $\sum_{j \in C_{\ell}} s_j > B$, and further, $\sum_{j \in C_{\ell}} s_j \underline{x}_j = \eps B$.

Given the current allocation $\mathbf{x}$, we find candidates $j, \ell \in C_{\ell}$ such that (i) $x_j \le 1 -\frac{\delta}{s_j}$; (ii) $x_{\ell} \ge \underline{x}_{\ell}+\frac{\delta}{s_\ell}$; and (iii) the following condition holds: 
\begin{equation}
	\label{eq:local}
	\frac{\partial \phi(\mx)}{\partial x_j \cdot s_j} >  \frac{\partial \phi(\mx)}{\partial x_{\ell} \cdot s_{\ell}} +\frac{\eps}{b}.
\end{equation}

Recall that $b$ is the budget for the overall problem, not for the subroutine. While such a pair of candidates $(j, \ell)$ exists, the algorithm increases $x_j$ by $\delta/s_j$, and decreases $x_{\ell}$ by $\delta/s_{\ell}$. Note that this update is feasible for the program, and preserves the cost of the allocation in $C_{\ell}$ at exactly $B$, and is hence feasible for the above program. The procedure stops when such a pair of candidates $(j,\ell)$ no longer exists.

\subsubsection{Analysis} 
In the analysis below, we assume the local optimum can be efficiently computed and present properties of this solution. The running time analysis is presented in \cref{sec:polyproof}.

We show the following result, which at a high level, extends the analysis of the PAV rule for binary additive utilities and unit sizes in~\cite{PetersS20} to the continuous setting, with submodular utilities and arbitrary sizes. This result will show as a corollary that \nw{}$(C, W, \{U_{i}\}, B)$ finds a $2$-approximate fractional core outcome. 

\begin{theorem}
	\label{thm:grad}
	Given a set of voters $W$ with utilities $\{u_i\}$ and multilinear extensions $\{U_i\}$, and a cost budget $B \in \left[\frac{\eps}{5m} b, b \right]$ on $C_{\ell}$, let $\mathbf{x}$ denote the solution to \nw{}$(C, W, \{U_{i}\}, B)$.  Suppose a subset $S \subseteq W$ of voters can choose an allocation $\mathbf{y} \in [0,1]^m$ with $\cost(\mathbf{y}) \le b$ such that $U_i(\mathbf{y}) > \theta U_i(\mathbf{x})$ for all $i \in S$, where $\theta \ge 1$. Then:
\[
|S| < \frac{|W|}{B (1 -\eps)}\cdot \frac{\cost(\mathbf{y})}{\theta-1-2\eps}.
\]
\end{theorem}
\begin{proof}
	Given the local optimum $\mathbf{x}$ with $\sum_{j \in C_{\ell}} s_j x_j = B$. Let $M_1 = \{j \in C \mid x_j +\frac{\delta}{s_j}\le y_j\}$ and $M_2= \{j \in C \mid x_j <  y_j<x_j+\frac{\delta}{s_j}\}$. Note that $M_1 \cup M_2 \subseteq C_{\ell}$, since for all $j \in C_s$, we have $x_j = 1 \ge y_j$. 
	
	Let $\mathbf{y'}$ be the allocation where $y'_j = y_j$ for $j\in M_1$, and $y'_j = x_j$ otherwise. Note that when moving from $\mathbf{x}$ to $\mathbf{y'}$, utility increase for $i \in S$ can only come from allocation increase of candidates in $M_1$. Further, note that the direction $\mathbf{y'} - \mathbf{x}$ is non-negative, since $y_j > x_j$ for $j \in M_1$ and $y'_j = x_j$ for $j \notin M_1$. 
	
	By applying \cref{lem:conm} to the multilinear function $U_i$ and observing that $\frac{\partial U_i}{\partial x_j} \ge 0$, we have 
	$$\sum_{j \in M_1} \frac{\partial U_i(\mathbf{x})}{\partial x_j}\cdot y_j + U_i(\mathbf{x}) \ge \sum_{j \in C} \frac{\partial U_i(\mathbf{x})}{\partial x_j}\cdot (y'_j-x_j) + U_i(\mathbf{x}) \ge U_i(\mathbf{y'}),$$ 
	where the last inequality follows from the concavity of $U_i$ in non-negative directions. 
	
	Consider now the allocation $\mathbf{z}$, where $z_j = y_j$ if $j \in M_1 \cup M_2$, and $z_j = x_j$ otherwise. Note that $\mathbf{z} \ge \mathbf{y}$, so that $U_i(\mathbf{z}) \ge U_i(\mathbf{y})$. Further, the direction $\mathbf{z} - \mathbf{y'}$ is non-negative, since $z_j = y_j \ge x_j = y'_j$ for $j \in M_2$, and $z_j = y'_j$ otherwise. Applying \cref{lem:conm} again gives
\[
U_i(\mathbf{z})\le U_i(\mathbf{y'})+\sum_{j} \frac{\partial U_i(\mathbf{y'})}{\partial y'_j}\cdot (z_j-y_j') =  U_i(\mathbf{y'})+\sum_{j\in M_2} \frac{\partial U_i(\mathbf{y'})}{\partial x_j}\cdot (y_j-y_j').
\]

Combining the relations above and observing that $U_i(\mathbf{z}) \ge U_i(\mathbf{y})$, we have
$$ \sum_{j \in M_1} \frac{\partial U_i(\mathbf{x})}{\partial x_j}\cdot y_j + U_i(\mathbf{x}) \ge U_i(\mathbf{y'}) \ge U_i(\mathbf{y}) - \sum_{j\in M_2} \frac{\partial U_i(\mathbf{y'})}{\partial x_j}\cdot (y_j-y_j').$$

By \cref{lem:gradsize}, we have $\frac{\partial U_i(\mathbf{y'})}{\partial x_j} \le 1$ for all $j \in M_2$. Further, $y_j - y'_j = y_j - x_j \le \frac{\delta}{s_j} \le \frac{\eps^6}{312m^5}$, since $s_j \ge \frac{\eps}{m} b$ and $\delta = \dstate$. Therefore, for any $i \in S$, we have:
$$ \sum_{j\in M_1} \frac{\partial U_i(\mathbf{x})}{\partial x_j}\cdot  y_j \ge U_i(\mathbf{y}) - U_i(\mathbf{x}) - \frac{\eps^6}{312m^4} > (\theta - 1) U_i(\mathbf{x}) - \frac{\eps^6}{312m^4},$$
where we have used that $U_i(\mathbf{y}) > \theta U_i(\mathbf{x})$ for all $i \in S$.

Note now that 
\begin{equation}
    \label{eq:U_lb}
    U_i(\mathbf{x})\ge u_i^{\max} \cdot \underline{x}_j=\frac{B \cdot \eps}{\sum_{j\in C_{\ell}} s_j}\ge \frac{\eps^2 \cdot b/m}{b \cdot m}=\frac{\eps^2}{m^2},
\end{equation}
where we have used that $s_j \le b$ for all $j \in C$, and $B \ge \frac{\eps}{5 m} b$. Combining the two inequalities above gives
\[
\sum_{j\in M_1} \frac{\partial U_i(\mathbf{x})}{\partial x_j}\cdot  y_j\ge  (\theta - 1-\eps)\cdot U_i(\mathbf{x}).
\]
Summing this inequality over all candidates in $S$, we have
\[
\sum_{i\in S} \sum_{j\in M_1} \frac{\partial U_i(\mathbf{x})}{\partial x_j}\cdot \frac{y_j}{U_i(\mathbf{x})}\ge|S| \cdot(\theta-1-\eps).
\]
	Since $\frac{\partial \phi_i(\mx)}{\partial x_j }=\frac{1}{U_i(\mathbf{x})}\cdot \frac{\partial U_i(\mathbf{x})}{\partial x_j}$, we have 
	\begin{align*}
	\frac{\sum_{j\in M_1}\left( y_j\cdot s_j \cdot \frac{\partial \phi(\mx)}{\partial x_j \cdot s_j}\right)}{\sum_{j\in M_1} y_j\cdot s_j} \ge \frac{\sum_{j\in M_1}\left( \sum_{i\in S}\frac{\partial \phi_i(\mx)}{\partial x_j \cdot s_j}\cdot y_j\cdot s_j\right)}{\sum_{j\in M_1} y_j\cdot s_j} &\ge \frac{\sum_{i \in S} \sum_{j\in M_1}\frac{\partial U_i(\mathbf{x})}{\partial x_j}\cdot \frac{y_j}{U_i(\mathbf{x})}}{\sum_{j\in M_1} y_j \cdot s_j}\\
	&\ge \frac{|S| \cdot (\theta - 1-\eps)}{\cost(\mathbf{y})}. 	    
	\end{align*}
	Therefore,
	\[
	\max_{j\in M_1} \frac{\partial \phi(\mx)}{\partial x_j \cdot s_j} > \frac{|S|\cdot (\theta - 1-\eps)}{\cost(\mathbf{y})} \ge \frac{|S|\cdot (\theta - 1-2\eps)}{\cost(\mathbf{y})}+\veps,
	\]
	where we have used that $\cost(\mathbf{y}) \le b$ and $|S| \ge 1$. Let $j_1$ denote the $j\in M_1$ that achieves the maximum in the previous inequality.
	
	Now, let $R_1 = \{j \in C_{\ell} \mid x_j \ge \underline{x}_j+\frac{\delta}{s_j}\}$. Applying \cref{lem:conm} along the positive direction $\mathbf{x}$ gives $\sum_{j \in C} x_j \cdot \frac{\partial U_i(\mathbf{x})}{\partial x_j}\le U_i(\mathbf{x})$. Since $\frac{\partial U_i(\mathbf{x})}{\partial x_j} \ge 0$ for all $j \in C$, we have
	$$ \sum_{j \in R_1} x_j \cdot \frac{\partial U_i(\mathbf{x})}{\partial x_j}\le U_i(\mathbf{x}). $$
	Summing this over all $i \in W$, we have:
	$$
	\sum_{i \in W} \sum_{j \in R_1} x_j \cdot s_j \cdot \frac{\partial \phi_i(\mx)}{\partial x_j \cdot s_j}  = \sum_{i \in W} \frac{\sum_{j \in R_1} x_j \cdot \frac{\partial U_i(\mathbf{x})}{\partial x_j}} {U_i(\mathbf{x})} \le \sum_{i \in W} \frac{U_i(\mathbf{x})}{U_i(\mathbf{x})} = |W|.
	$$
	Therefore,
	$$ \frac{\sum_{j \in R_1} x_j \cdot s_j \cdot \frac{\partial \phi(\mx)}{\partial x_j \cdot s_j}}{\sum_{j \in R_1} x_j \cdot s_j} = \frac{\sum_{i\in W} \sum_{j \in R_1} x_j \cdot s_j \cdot \frac{\partial \phi_i(\mx)}{\partial x_j \cdot s_j}}{\sum_{j \in R_1} x_j \cdot s_j}\le \frac{|W|}{\sum_{j \in R_1} x_j \cdot s_j} \le  \frac{|W|}{B(1-\eps)}.$$
	The last inequality holds since $\sum_{j \in C_{\ell} \setminus R_1} x_j \cdot s_j\le \sum_{j\in C_{\ell} \setminus R_1} \frac{\delta}{s_j}\cdot s_j \le m \delta \le \eps B$, where we have used that $\delta = \dstate$ and $B \ge \frac{\eps}{ 5 m} b$. This implies  $\sum_{j \in R_1} s_j x_j \ge B( 1- \eps).$
	
	Therefore,
	$$ \min_{j \in R_1} \frac{\partial \phi(\mx)}{\partial x_j \cdot s_j} \le \frac{|W|}{B(1-\eps)}.$$
	Let $j_2$ denote the $j \in R_1$ that achieves this minimum. 
	
	Note that $x_{j_1} < y_{j_1}$ and $x_{j_2} > \underline{x}_{j_2}$ by assumption. Further, since $\mathbf{x}$ is locally optimal, \cref{eq:local} cannot hold for $j = j_2$ and $\ell = j_1$. This means $\frac{\partial \phi(\mx)}{\partial x_{j_1}\cdot s_{j_1}} \le \frac{\partial \phi(\mx)}{\partial x_{j_2}\cdot s_{j_2}}+\veps$. Therefore, we have
	$$\frac{|S| \cdot (\theta - 1-2\eps)}{\cost(\mathbf{y})} < \frac{\partial \phi(\mx)}{\partial x_{j_1}\cdot s_{j_1}}-\veps \le \frac{\partial \phi(\mx)}{\partial x_{j_2}\cdot s_{j_2}}  \le \frac{|W|}{B(1-\eps)}.$$
	Rearranging completes the proof.
\end{proof}

As an aside, the following corollary is immediate if all one wants is an approximate fractional core on the utilities $\{U_i\}$ and budget $b$. Note that the proof trivially extends to {\em any} continuous, concave utilities $\{U_i\}$, since these are concave in positive directions.

\begin{corollary}
	The solution $\mathbf{x}$ of \nw{}$(C, V, \{U_{i}\}, b)$ is in a $2$-approximate fractional core for {\em any} continuous concave utilities $\{U_i\}$.
\end{corollary}
\begin{proof}
	Setting $\theta = 2$ and $\varepsilon = 0$ in \cref{thm:grad}, if there exists an $S$ and a solution $\mathbf{y}$ such that $U_i(\mathbf{y}) > 2 U_i(\mathbf{x})$ for all $i \in S$, then $|S| < \frac{|W|}{b}\cdot \cost{}(\mathbf{y})$. This contradicts \cref{def:core2} and thus completes the proof.
\end{proof}

\subsection{Randomized Rounding and Satisfied Voters}
We now present a randomized rounding scheme on the fractional solution constructed above. Since the utilities of voters are different but the allocation is the same, we will need the rounding of the fractional allocation to be {\em oblivious} to the utilities. The natural approach is therefore randomized rounding. Clearly, if we randomly round, the expected utility of a voter in the integral committee is at least that in the fractional core. However, the same cannot be said for an arbitrary \emph{realization} of the randomization: there can be many voters having lower-than-expected utilities, and they can form a coalition and deviate. Further, the size of the resulting committee can exceed the bound $B$.  Therefore, we cannot directly argue that the resulting integral committee will lie in the core. 

Our contribution in this section is to show a rounding procedure that ensures that with constant probability, a constant fraction of voters obtain at least a constant fraction of the utilities they obtained in the fractional solution. We will use this as a subroutine in our main procedure in \cref{sec:proof} in order to argue that these voters will not form part of any deviating coalition. 


\subsubsection{Rounding Procedure}
To begin with, in the sequel, we will distinguish between items that are fractionally and integrally allocated as follows:

\begin{definition}[Fully and Fractionally Allocated Items]
	An item $j \in C$ in the fractional solution is \emph{fully allocated} if $x_j \ge 1$; we call it \emph{fractionally allocated} if $x_j \in (0, 1)$. Note that all items in $C_s$ are fully allocated.
\end{definition}




The rounding procedure has a parameter $\kappa \le 1$, and executes as shown in \cref{alg:rand}.

\begin{algorithm}[htbp]
	\caption{Randomized Rounding of Fractional Allocation \label{alg:rand} $\mathbf{x}$}
	\begin{algorithmic}[1]
		\Statex
		\Function{\idep}{$C, W, \{U_{i}\}, B$} 
		\Let{$\mathbf{x}$} {\nw{}($C, W, \{U_{i}\}, \kappa B$)}
		\Let{$O$}{$C_s$} 
		\Let{$T_2$}{ $\{j \in C_{\ell} \mid s_j \le \kappa B\}$}
		\State Include $j \in T_2$ in $O$ independently with probability $\min(1,x_j)$
		\State \Return{$O$}
		\EndFunction
	\end{algorithmic}
\end{algorithm}

\subsubsection{Satisfied Voters}
We now define what it means for a voter to be {\em satisfied}, and present a bound on the number of satisfied voters after the procedure \idep{} is applied.

\begin{definition}[$\gamma$-Satisfied]
	\label{def:gammaSatisfied}
	Let $\mathbf{x}$ be the fractional solution of \nw{}$(C, W, \{U_i\}, \kappa B)$, and $O$ be the resulting integer committee found by \idep{}$(C, W, \{U_{i}\}, B)$. For $\gamma \ge 2$, we say that voter $i$ is \emph{$\gamma$-satisfied} by $O$ w.r.t. solution $\mathbf{x}$ if there is a $q \in C$ such that $u_i(O \cup \{q\}) \ge \frac{U_i(\mathbf{x})}{\gamma}$.
\end{definition}

Our main result in this section is the following theorem.

\begin{theorem}[Constant Fraction of Constant-Satisfied Voters for Submodular Utilities]
	\label{thm:nw_satisfied}
	Given the fractional solution $\mathbf{x}$ produced by \nw{}$(C, W, \{U_{i}\}, \kappa B)$ where $|W| = n'$, there is a integral committee $O$ produced by \idep{}$(C, W, \{U_{i}\}, B)$ with at least $(1 - \beta)n'$ $\gamma$-satisfied voters, where 
	\begin{equation}
	\label{eq:beta}
	    \beta = \left( \kappa e^{1-\kappa} \right)^{\frac{1}{\kappa}} +  (\gamma -1) e^{2-\gamma}.
	\end{equation}
	Further, this committee is feasible, so that $\sum_{j \in O \cap C_{\ell}} s_j \le B$.
\end{theorem}

For the proof, we will need the following lemma is the analog of \cref{lem:chernoff} for submodular functions.

\begin{lemma}[Lower Tail Bound, \citep{ChekuriVZ10}]
	\label{lem:chernoff2}
	Assume $f(\{j\}) \in [0,1]$ for all $j \in C$, and let $\mu_0 = F(\mathbf{x})$, where $F$ is the multilinear extension of $f$. Let $O$ denote the outcome of independent randomized rounding applied to $\mathbf{x} \in [0,1]^m$.  Then for any $\rho > 1$:
	$$ \Prob{}{f(O) \le \frac{\mu_0}{\rho}} \le \left(\rho e^{1-\rho} \right)^{\frac{\mu_0}{\rho}} .$$
\end{lemma}

\begin{proof}[Proof of \cref{thm:nw_satisfied}]
	We show that for each voter $i$, the probability that $i$ is $\gamma$-satisfied is at least $1 - \beta$. This implies the number of $\gamma$-satisfied voters is at least $(1 - \beta) n'$ in expectation, and thus it is at least $(1 - \beta) n'$ for some realization produced by \idep{}.
	
	Let $T_1 = \{ j \in C_{\ell} \mid s_j \ge \kappa B\}$ and $T_2 = C \setminus T_1$. Fix voter $i$ with utility function $u_i$, and multilinear extension $U_i$. 
	
	Let $G_1$ denote the multilinear extension restricted only to candidates in $T_1$, that is,
	$$ G_1(\mathbf{x}) =  \sum_{T \subseteq T_1} u_i(T) \left(\prod_{j \in T} x_j \right) \left(\prod_{j \in T_1\setminus T} (1-x_j) \right).$$
	and similarly, $G_2$ denote the multilinear extension restricted to $T_2$, that is,
	$$ G_2(\mathbf{x}) =  \sum_{T \subseteq T_2} u_i(T) \left(\prod_{j \in T} x_j \right) \left(\prod_{j \in T_2\setminus T} (1-x_j) \right).$$
	By sub-additivity of $u_i$, we have $U_i(\mathbf{x})  \le G_1(\mathbf{x}) + G_2(\mathbf{x})$, and by monotonicity of $u_i$, we have $U_i(\mathbf{x}) \ge \max\left(G_1(\mathbf{x}), G_2(\mathbf{x}) \right)$.
	
	We now split the analysis into several cases:
	
	\noindent\textbf{Case (1):} Suppose $G_1(\mathbf{x}) \ge \frac{1}{\gamma} U_i(\mathbf{x})$.  In this case, we make two observations. First,
	$$ G_1(\mathbf{x}) \le \sum_{T \subseteq T_1} \left(\sum_{j \in T} u_i(\{j\})\right) \left(\prod_{j \in T} x_j \right) \left(\prod_{j \in T_1\setminus T} (1-x_j) \right) \le \sum_{j \in T_1} x_j u_i(\{j\}).$$
	Further, since $\sum_{j \in T_1} s_j x_j \le \kappa B$, and since $s_j \ge \kappa B$ for all $j \in T_1$, we have $\sum_{j \in T_1} x_j \le 1.$ Putting these together, we have
	$$ G_1(\mathbf{x}) \le \max_{j \in T_1} u_i(\{j\}).$$
	Therefore, using $j^* = \mbox{argmax}_{j \in T_1} u_i(\{j\})$ as additament makes $i$ achieve utility at least $1/\g$ fraction of $U_i(\mathbf{x})$. Therefore, $i$ is $\gamma$-satisfied just by the additament.
	
	\noindent\textbf{Case (2):} We now assume $G_1(\mathbf{x}) \le \frac{1}{\g} U_i(\mathbf{x})$ and thus $G_2(\mathbf{x}) \ge (1-1/\g) U_i(\mathbf{x})$. Suppose there exists $\ell \in T_2$ such that $u_i(\{\ell\}) \ge \frac{1}{\g - 1} G_2(\mathbf{x})$. Then, just using $\ell$ as additament causes the utility of $i$ to be at least $1/\g$ fraction of $F_i(\mathbf{x})$ and voter $i$ is $\gamma$-satisfied by the additament $\ell$.
	
	\noindent\textbf{Case (3):} In the final case, we have $G_2(\mathbf{x}) \ge (1-1/\g) U_i(\mathbf{x})$, and we assume that for all $j \in T_2$, $u_i(\{\ell\}) < \frac{1}{\g - 1} G_2(\mathbf{x})$. Let $O' = O \cap C_{\ell}$.
	
	First note that $\Ex{}{\sum_{j \in O'} s_j} \le \kappa B$, since $O'$ is the result of randomized rounding of $\mathbf{x}$ and since $\sum_{j \in C_{\ell}} s_j x_j \le \kappa B$. Further, $O' \subseteq T_2 \cap C_{\ell}$, and for all $j \in T_2 \cap C_{\ell}$, we have $s_j \le \kappa B$. Therefore, a standard application of Chernoff bounds yields:
	\[
	\Prob{}{\sum_{j \in O'} s_j > B} \le \left( \kappa e^{1-\kappa} \right)^{\frac{1}{\kappa}}.
	\]
	Further, note that for all $j \in T_2$, we have that the marginal $u_i(\{\ell\}) < \frac{1}{\g - 1} G_2(\mathbf{x})$. Applying \cref{lem:chernoff2} with $\rho = \g -1$, we have:
	$$ \Prob{}{u_i(O) \le \frac{G_2(\mathbf{x})}{\g-1}} \le (\g -1) e^{2-\g}.$$
	Therefore, by union bounds, with probability at least $1 - \beta$, we have both events: (1) The committee $O'$ is feasible for size $B$; and (2) $u_i(O) \ge \frac{G_2(\mathbf{x})}{\g-1} \ge \frac{U_i(\mathbf{x})}{\g}$, implying $i$ is $\gamma$-satisfied by $O$. This completes the proof.
\end{proof}

The following corollary  shows that this step can be implemented in polynomial time.

\begin{corollary}
\label{cor:beta}
    For $\beta$ defined in \cref{eq:beta}, $n' = |W|$ and any $\eps \in (0,1)$, a committee $O \subseteq C$ with $(1-\beta - \eps)n'$ $\gamma$-satisfied voters can be computed with probability $1 - \frac{1}{\mbox{poly}(m,n)}$ in time $\mbox{poly}(n,m,1/\eps)$.
\end{corollary}
\begin{proof}
    Let $X$ be the random variable indicating the number of $\gamma$-satisfied voters returned by \idep{}$(C, W, \{U_{i}\}, B)$. By \cref{thm:nw_satisfied}, we have $\Ex{}{X} = (1-\beta)n'$. Further, $X \le n'$. Therefore, $\Pr[X \ge (1-\beta - \eps)n] \ge \frac{\eps}{\beta + \eps} \ge \frac{\eps}{2}$. A standard application of sampling bounds now completes the proof.
\end{proof}

\subsection{The Constant Approximation to the Core}
\label{sec:proof}
In this section, we design an algorithm for the overall problem, and prove that it returns a committee in an $O(1)$-core. The algorithm will repeatedly construct fractional solutions and round them, using the algorithms developed above as subroutines. The analysis critically requires the local optimality property of the Nash Welfare objective, captured in \cref{thm:grad}.

\subsubsection{Algorithm}
\cref{alg:1} finds an approximate core solution in the following way: It iteratively computes the fractional local optimum to Nash Welfare on the remaining voters with a scaled-down budget, rounds it, eliminates voters that are $\gamma$-satisfied with respect to the solution at this iteration, scales down the budget again, and iterates on the remaining voters with the smaller budget. The scaling parameter $\omega$ for the budget will be determined later. The overall structure of the algorithm is similar to that in~\citep{DBLP:conf/stoc/JiangMW20}, though the details of constructing the fractional solution, and the resulting proof of correctness are entirely different.

\begin{algorithm}
	\caption{Iterative Rounding of the Nash Welfare Solution \label{alg:1}}
	\begin{algorithmic}[1]
		\Statex
		\Function{IterRound}{$C, V, \{u_{ij}\}, b$}
		\Let{$t$}{$0$} \Let{$V_0$}{$V$} \Let{$T^*$}{$C_{s}$} \Let{$b_0$}{$(1-\eps)(1 - \omega)b$}
        \While{$b_t \ge \frac{\eps}{m} b$}
		\Let{$\mathbf{x}_t$} {\nw{}($C, V_t, \{U_{i}\},\kappa b_t$)}
		\Let{$O_t$} {Solution of \idep{}$(C, W, \{U_{i}\}, b_t)$ that satisfies \cref{cor:beta}} \label{algline:round}
		\Let{$W_t$}{Voters in $V_t$ that are $\gamma$-satisfied by $O_t$ with respect to the solution $\mathbf{x}_t$}
		\Let{$V_{t + 1}$}{$V_t \setminus W_t$}
		\Let{$T^*$}{$T^* \cup O_t$}
		\Let{$b_{t + 1}$}{$\omega b_t$}
		\Let{$t$}{$t + 1$}
		\EndWhile
		\State \Return{\(T^*\) as the final integer solution}
		\EndFunction
	\end{algorithmic}
\end{algorithm}

\subsubsection{Analysis}
\label{sec:mainanalysis}
Let $n = |V|$, and recall that $b$ is the initial budget. We will show that the solution $T^*$ lies in the $\alpha$-core for the set $V$ of voters with size constraint $b$. First note that since $\sum_{j \in C_s} s_j \le \eps b$, and since each $O_t$ is feasible for budget $b_t$, we have: 
\[
\sum_{j \in T^*} s_j \le \eps b + \sum_t \sum_{j \in O_t} s_j \le \eps b + \sum_t b_t = \eps b + (1-\eps) (1-\omega) \sum_{t \ge 0} \omega^t b \leq b.
\]
Therefore, the solution $T^*$ is feasible for the size $b$.

First note that when $b_t < \frac{\eps}{m} b$, since all items in $C_{\ell}$ has $s_j > \frac{\eps}{m}b$, we have $\sum_{j \in C_{\ell}} x_{jt} < 1$.  Since $C_s \subseteq T^*$, for any fractional solution $\mathbf{x_t}$, we have $\max_{j \in C} u(\{j\}) \ge U_i(\mathbf{x_t})$, so that all voters are $1$-satisfied. This implies $V_t = \varnothing$ at termination, so that any voter $i \in V$ belongs to $W_{t'}$ for some $t'.$

For the purpose of contradiction, we assume the resulting solution $T^*$ is not in the $\alpha$-core (\cref{def:alphaCore}), where $\alpha \ge 1$ is a quantity we will determine later. Let $S$ denote the set of voters that deviate, and let $A$ denote the set of items they deviate to. We have $ \sum_{j \in A} s_j \le \frac{|S|}{n} \cdot b$, and $A$ provides an $\alpha$-factor larger utility to voters in $S$ even after including any additament.

Consider the voters in $W_t$, and let $S_t = S \cap W_t$. These voters are $\gamma$-satisfied by $O_t$ with respect to the fractional solution $\mathbf{x}_t$ (\cref{def:gammaSatisfied}). Therefore, if $i \in S_t$ deviates to $A$ to obtain an $\alpha$-factor larger utility, it must be that $U_i(A) \ge \frac{\alpha}{\gamma} \cdot U_i(\mathbf{x}_t)$. Let $\theta = \frac{\alpha}{\gamma}$. We will assume $\theta > 1$ below.

Using \cref{thm:grad}, since $\mathbf{x}_t$ is a local optimum, using $\mathbf{y}$ as the set $S_t$, and observing that all agents $i \in S_t$ have $U_i(\mathbf{y}) > \theta U_i(\mathbf{x}_t)$, where $\theta = \frac{\alpha}{\gamma}$, we have:
$$|S_t| \le \frac{n_t}{\kappa b_t\cdot (1-\eps)} \cdot \frac{\sum_{j \in A} s_j}{\theta - 1-2\eps}.$$

Summing this over all $t$, and using $ \sum_{j \in A} s_j \le \frac{|S|}{n} b$, we have:
\[
|S| = \sum_t |S_t| \le \frac{\sum_{j \in A} s_j}{(\theta - 1-2\eps)(1-\eps)} \cdot \sum_t \frac{n_t}{\kappa b_t} \le \frac{|S|}{n} \cdot \frac{b}{(\theta - 1-2\eps)(1-\eps)} \cdot \sum_t \frac{n_t}{\kappa b_t}.
\]
Therefore, for a blocking coalition to exist, we need:
\begin{equation} \label{eq:block1}
	\frac{b}{n} \cdot \sum_t \frac{n_t}{\kappa b_t} \ge (\theta - 1-2\eps)(1-\eps) = \left(\frac\alpha\gamma - 1-2\eps\right )(1-\eps).
\end{equation}
We will now set the parameters $\omega, \gamma, \alpha$ so that the above inequality is false. First note by \cref{cor:beta} that $n_{t+1} \le (\beta + \eps) n_t$ where $\beta$ satisfied \cref{eq:beta}. Further, $b_{t+1} = \omega b_t$. Therefore, 
\[
\frac{n_{t+1}}{b_{t+1}} \le \frac{\beta + \eps}{\omega} \cdot \frac{n_t}{b_t}
\]
with $\frac{n_0}{b_0} = \frac{n}{(1-\omega)(1-\eps) b}$. Therefore,
\begin{equation} \label{eq:block2}
	\frac{(1-\eps)b}{n} \cdot \sum_t \frac{n_t}{\kappa b_t} \le \frac{1}{(1-\omega) \kappa} \cdot \sum_{t \ge 0} \left( \frac{\beta + \eps}{\omega} \right)^t = \frac{\omega}{(1-\omega)(\omega - \beta - \eps) \kappa}.
\end{equation}
Combining \cref{eq:block1} and \cref{eq:block2}, for a blocking coalition to exist, we need
\[
\alpha \le \frac{\omega \gamma}{\kappa(1-\omega)\left(\omega - (\g - 1)e^{2-\g} - \left( \kappa e^{1-\kappa} \right)^{\frac{1}{\kappa}} - \eps\right)(1-\eps)^2}+(1+2\eps)\cdot \gamma.
\]
For an $\alpha$ slightly larger than the right-hand side, a blocking coalition will therefore not exist. Then The right-hand side of the above inequality is approximately minimized when $\omega = 0.23$, $\gamma = 7.435$, $\kappa = 0.21$ and $\eps \rightarrow 0$, yielding $\alpha < 67.37$.\footnote{Note that for this choice of $\kappa$, the subroutine \nw{}($C, V_t, \{U_{i}\},\kappa b_t$) is run with budget at least $\kappa b_t \ge \kappa \frac{\eps}{m} b \ge \frac{\eps}{5m} b$, so that the precondition of \cref{thm:grad} holds.} This finally yields the following theorem; the only missing detail is the running time of local search in \cref{sec:nash}, which we address in \cref{sec:polyproof}.

\main*

\subsection{The Running Time of Local Search}
\label{sec:polyproof}
We now show that \cref{alg:1} runs in polynomial time. This requires showing that the local search procedure in \cref{sec:nash} runs in polynomial time; the rest of the steps in \cref{alg:1} can easily be implemented efficiently. We will show this in two parts: The partial derivatives $\frac{\partial \phi_i(\mx)}{\partial x_j}$ can be approximately computed efficiently; and the number of iterations (finding the candidate pair $(j,\ell)$ and updating the allocations) performed by the local search procedure is polynomially bounded.


\subsubsection{Estimating the Gradient of the Nash Welfare Objective}
We first show how to estimate the derivative of $\phi(\mx)$ in the procedure in \cref{sec:nash}. First, the estimation procedure for $U_i$ and its derivative is the same as that in \citet{CalinescuCPV11}: Each time we compute $U_i(\mathbf{x})$, we pretend $\mathbf{x}$ is a product distribution over candidates, and sample $H$ times it. Denote the samples as the committees $O_h \sim \mathbf{x}$ where $h \in \{1,2,\ldots,H\}$. We calculate the quantity $\hat{U_i}(\mathbf{x})=\frac{1}{H} \cdot \sum_{h=1}^H u_i(O_h)$ as the estimation of $U_i(\mathbf{x})$. The following lemma gives a bound on the additive error of estimation on $U_i(\mathbf{x})$ if $H$ is sufficiently large.

\begin{lemma}[\citet{CalinescuCPV11}] $\Pr\left[\left\vert \hat{U}_i(\mathbf{x})-U_i(\mathbf{x}) \right\vert>\Delta\right]<2e^{-\frac{\Delta^2\cdot H}{m^2}}.$
\end{lemma}

Since the derivative of $U_i$ is given in \cref{lem:gradsize}, its error is bounded by estimating the multilinear function twice. This yields the following lemma:

\begin{lemma}
\label{lem:eval}
    In the local search procedure in \cref{sec:nash}, suppose the total number of times the multilinear function $U_i$ and its derivative is evaluated is $N$. Then, using $\mbox{poly}\left(N, m, \frac{1}{\Delta} \right)$ samples, the additive error in each estimate is bounded by $\Delta$ with probability $1 - \frac{1}{\mbox{poly}(m,N)}$.  
\end{lemma}

In the sequel, we will set $\Delta=\destat$. We denote the estimated derivative of the multilinear function as $\frac{\ph U_i(\mathbf{x})}{\ph x_i}$, and the estimated derivative of $\phi_i(\mx)$ and $\phi(\mx)$  as $\frac{\ph\phi_i(\mx)}{\ph x_j}=\frac{1}{\hat{U}_i}\cdot \frac{\ph U_i(\mathbf{x})}{\ph x_j}$ and $\frac{\ph \phi(\mx)}{\ph x_j}=\sum_{i=1}^n \frac{\ph\phi_i(\mx)}{\ph x_j}$ respectively.  We have the following lemma for bounding the error of estimating $\frac{\partial\phi(\mx)}{\partial x_j}\cdot {\frac1{s_j}}$. For this proof and subsequent ones, we crucially need that the Nash Welfare program in \cref{sec:nash} lower bounds the allocations as $x_j \ge \underline{x}_j$.

\begin{lemma}
\label{lem:gradbound}
     Let $\Delta=\destat$, and suppose $x_j\ge \underline{x}_j$ for all candidates $j \in C_{\ell}$, where $\underline{x}_j$ is as defined in \cref{sec:nash}. If for any $j \in C_{\ell}$ and all $i \in V$, we have: $\left\vert\hat{U}_i(\mathbf{x})-U_i(\mathbf{x})\right\vert< \Delta$ and  $\left\vert \frac{\ph U_i(\mathbf{x})}{\ph x_j} - \frac{\partial U_i(\mathbf{x})}{\partial x_j}\right\vert< \Delta$, then it holds that $\frac{1}{s_j} \cdot  \left\vert\frac{\ph\phi(\mx)}{\ph x_j} -\frac{\partial \phi(\mx)}{\partial x_j}\right\vert\le \frac{\eps}{8b}.$
\end{lemma}
\begin{proof}
We have 
\begin{align*}
    \left\vert\frac{\ph\phi_i(\mx)}{\ph x_j}-\frac{\p\phi_i(\mx)}{\p x_j}\right\vert&= \left\vert\frac{1}{\hat{U}_i(\mathbf{x})}\cdot \frac{\ph U_i(\mathbf{x})}{\ph x_j}-\frac{1}{\ux}\cdot \frac{\p\ux}{\p x_j}\right\vert\\
    &\le \frac{1}{U_i(\mathbf{x})-\Delta}\cdot \left(\frac{\partial U_i(\mathbf{x})}{\partial x_j}+\Delta\right)-\frac{1}{\ux}\cdot \frac{\p\ux}{\p x_j}\\
    &\le \left \vert\frac{\Delta \cdot \left(\ux+\frac{\p\ux}{\p x_j}\right)}{(\ux-\Delta)^2}\right\vert.
\end{align*}
By \cref{lem:gradsize}, we have $\frac{\p\ux}{\p x_j}\le 1$. If $\Delta\le \frac{\eps^2}{2m^2}$, then by \cref{eq:U_lb} we have $2 \Delta\le \ux$, thus:
\begin{align*}
    \left\vert\frac{\ph\phi_i(\mx)}{\ph x_j}-\frac{\p\phi_i(\mx)}{\p x_j}\right\vert&=\left \vert\frac{\Delta \cdot \left(\ux+\frac{\p\ux}{\p x_j}\right)}{(\ux-\Delta)^2}\right\vert \le \left \vert\frac{\Delta \cdot \left(\ux+1\right)}{\ux^2/4}\right\vert \le \left \vert\frac{\Delta \cdot 2}{\frac{\eps^4}{m^4}/4}\right\vert \le 
   \left \vert\frac{8\Delta \cdot m^4}{\eps^4}\right\vert.
\end{align*}
 
Since $s_j \ge \frac{\eps}{m} b$ for $j \in C_{\ell}$, we have 
\[
\frac1{s_j} \cdot \left\vert\frac{\partial \phi(\mx)}{\partial x_j} -\frac{\ph \phi(\mx)}{\ph x_j}\right\vert\le \frac{1}{s_j} \sum_{i=1}^n \left\vert\frac{\ph\phi_i(\mx)}{\ph x_j}-\frac{\p\phi_i(\mx)}{\p x_j}\right\vert\le \left \vert\frac{8\Delta \cdot m^4}{\eps^4}\right\vert \cdot n \cdot \frac{m}{b\cdot \eps}\le \frac{8\Delta \cdot m^5\cdot n}{\eps^5\cdot b}=\frac{\eps}{8b}. \qedhere
\]
\end{proof}
 
\subsubsection{Number of Iterations in Local Search}
The local search procedure in \cref{sec:nash} iteratively finds a pair of candidates $(j,\ell)$ such that 
\begin{equation}
\label{eq:localapprox}
    \frac{\ph \phi(\mx)}{\ph x_j}\cdot \frac{1}{s_j} >  \frac{\ph \phi(\mx)}{\ph x_{\ell}} \cdot \frac{1}{s_{\ell}} +\frac{3\eps}{4b}.
\end{equation}
Each time it finds such a pair, it increases $x_j$ by $\delta\cdot \frac{1}{s_{j}}$ and decreases $x_{\ell}$ by $\delta\cdot \frac{1}{s_{\ell}}$. Note that at stopping, \cref{eq:localapprox} does not hold, which by \cref{lem:gradbound} implies \cref{eq:local} also does not hold with high probability. Therefore, the termination condition in \cref{sec:nash} is satisfied.

In order to bound the number of iterations of this procedure, we first show that the first and second partial derivatives of $\phi$ are bounded from above. 

\begin{lemma}
\label{lem:bound}
If $\phi_i(\mx)=\log U_i(\mathbf{x})$ and $x_j\ge \underline{x}_j$ for all candidates $j \in C_{\ell}$, where $\underline{x}_j$ is as defined in \cref{sec:nash}, then we have the following bounds for all $j,k \in C_{\ell}$ and $i \in W$:
\begin{gather*}
 0 \le \frac{\partial \phi_i(\mx)}{\partial x_j} = O\left(\frac{m^2}{\eps^2}\right) \qquad \mbox{and} \qquad \left\vert \frac{\partial^2\phi_i(\mx)}{\partial x_j \partial x_k}\right\vert = O\left(\frac{m^4}{\eps^4}\right).
\end{gather*}
\end{lemma}
\begin{proof}
Since$\phi_i(\mx)= \log U_i(\mathbf{x})$, using \cref{lem:gradsize}, we have $\frac{\partial \phi_i(\mx)}{\partial x_j}\le \frac{1}{U_i}. $ 
Combining this with \cref{eq:U_lb} gives $\frac{\p\phi_i(\mx)}{\p x_j}\le \frac{5 m^2}{\eps^2}$. 

We next bound the second order derivatives as follows. Here, $\mathbf{x_{-j,k}}$ is $\mathbf{x}$ with the $j^{\text{th}}$ and $k^{\text{th}}$ dimension removed, and $T \sim \mathbf{x_{-j,k}}$ means that $T$ is chosen by including $\ell \in C \setminus \{j,k\}$ independently with probability $x_{\ell}$. The first inequality below uses \cref{eq:U_lb}.
    \begin{align*}
       \left\vert\frac{\partial^2 \phi_i(\mx)}{\partial x_j\partial x_k}\right\vert&= \left\vert-\frac{1}{U_i(\mathbf{x})^2}\cdot \frac{\partial U_i(\mathbf{x})}{\partial x_j} + \frac{1}{U_i(\mathbf{x})}\cdot \frac{\partial^2 U_i(\mathbf{x})}{\partial x_j \partial x_k}\right\vert\\
       &\le \frac{25 m^4}{\eps^4}+\left\vert \frac{5m^2}{\eps^2}\cdot \sum_{T\subseteq C\setminus\{j,k\}} \Pr_{T \sim \mathbf{x_{-j,k}}} [T] \cdot  \left(u_i(T\cup\{j,k\})-u_i(T\cup\{j\})-u_i(T\cup\{k\})+u_i(T)\right) \right\vert\\
        &\le \frac{25 m^4}{\eps^4} + \frac{5 m^2}{\eps^2}\cdot \sum_{T\subseteq C\setminus\{j,k\}} \Pr_{T \sim \mathbf{x_{-j,k}}} [T]  \cdot \Big(u_i(T\cup\{k\})-u_i(T)-u_i(T\cup\{j,k\})+u_i(T\cup\{j\})\Big) \\ 
        &\le \frac{25 m^4}{\eps^4} + \frac{5 m^2}{\eps^2}\cdot \sum_{T\subseteq C\setminus\{j,k\}} \Pr_{T \sim \mathbf{x_{-j,k}}} [T]  \cdot 1 \le \frac{26 m^4}{\eps^4}. \qedhere
    \end{align*}
\end{proof}

The next lemma now follows from a standard application of first order Taylor approximation. We assume the local search procedure in \cref{sec:nash} iteratively finds a pair of candidates $(j,\ell)$ such that 
\cref{eq:localapprox} holds. 

\begin{lemma}
\label{lem:stepbound}
    Assuming all the estimates on $\frac{\p\phi(\mx)}{\p x_j}$ during execution of local search are within $\pm\frac{\eps}{8b}$ of the true values, the total number of iterations is $\mbox{poly}\left(m,n,\frac{1}{\eps}\right)$.
\end{lemma}
\begin{proof}
Since all estimates of $\frac{\p\phi(\mx)}{\p x_j}$ during execution of local search are within $\pm\frac{\eps}{8b}$ of the true values, if \cref{eq:localapprox} holds for $j,k \in C_{\ell}$, then:
$$ \frac{\p \phi(\mx)}{\p x_j}\cdot \frac{1}{s_j} >\frac{\ph \phi(\mx)}{\ph x_j}\cdot \frac{1}{s_j}-\frac{\eps}{8b} \ge \frac{\ph \phi(\mx)}{\ph x_{\ell}} \cdot \frac{1}{s_{\ell}}+\frac{3\eps}{4b}-\frac{\eps}{8b}\ge \frac{\p \phi(\mx)}{\p x_{\ell}} \cdot \frac{1}{s_{\ell}} +\frac{\eps}{2b}. $$
We lower bound the increase in $\phi(\mx)$ when $\mathbf{x} = (\mathbf{x_{-j,k}},x_j, x_k)$ moves to  $\mathbf{x'} = \left(\mathbf{x_{-j,k}}, x_j+\frac{\delta}{s_j},x_\ell- \frac{\delta}{s_\ell} \right)$.
By Taylor approximation, we have:
$$ \phi(\mathbf{x'}) - \phi(\mathbf{x}) \ge (\mathbf{x'} - \mathbf{x}) \cdot \nabla \phi(\mathbf{x}) - \sum_{r=1}^m (x'_r - x_r)^2 \frac{M}{2} $$
where $M$ is an upper bound on the absolute value of the second derivatives of $\phi$. Since $M = O\left(\frac{m^4}{\eps^4}\right)$ by \cref{lem:bound}, this implies:
$$ \phi(\mathbf{x'}) - \phi(\mathbf{x}) \ge \delta \cdot \left(\frac{\p \phi(\mx)}{x_j}\cdot \frac1{s_j}-\frac{\p\phi(\mx)}{x_\ell}\cdot \frac1{s_\ell}\right) - \frac{26m^4}{\eps^4}\cdot \frac{3\delta^2}{s_{\min}^2},  $$
where $s_{\min} = \min_{\ell \in C_{\ell}} s_{\ell} \ge \frac{\eps}{m} b$ and $\delta = \dstate$. Plugging these values in and simplifying, we have:
$$ \phi(\mathbf{x'}) - \phi(\mathbf{x}) = \Omega\left( \frac{\eps \delta}{b} \right) = \mbox{poly}\left(\eps, \frac{1}{n}, \frac{1}{m}\right).$$
Since $\phi(\mx)$ lies in $\left(\log \left(\frac{\eps^2}{5m^2}\right), n\log m\right)$, the number of iterations is  $\mbox{poly}\left(m,n,\frac{1}{\eps}\right)$.
\end{proof}

Combining the previous lemmas, the following theorem is now immediate by union bounds, where we assume $\eps > 0$ is a small enough constant. This completes the proof of \cref{thm:main}.
\begin{theorem}
    With probability at least $1-\dfrac{1}{\mbox{poly}(n,m,\frac{1}{\eps})}$, the local search algorithm in \cref{sec:nash} has running time $\mbox{poly}(n,m,\frac{1}{\eps})$. 
\end{theorem}
\begin{proof}
    By \cref{lem:stepbound}, if we estimate all the $U_i(\mathbf{x})$ and $\frac{\p \ux}{\p x_j}$ within $\pm \Delta$ additive error, the algorithm ends in $\poly(n,m,\frac1\eps)$ iterations. Since within each iteration we estimate $U_i$ and $\frac{\p \ux}{\p x_j}$ for $\mbox{poly}(m)$ times, the total number of evaluations is bounded by $N=\mbox{poly}(n,m,\frac1\eps)$. Then by \cref{lem:eval}, we need at most $\mbox{poly}(N,m,\frac{1}{\Delta})$ samples for all additive errors to be bounded by $\pm \Delta$ with probability $1-\frac{1}{\mbox{poly}(N)}$. Since $\Delta=\destat=\frac{1}{\mbox{poly}\left(n,m,\frac1\eps\right)}$ and the number of iterations is bounded by $\mbox{poly}(n,m,\frac1\eps)$, the total running time including sampling steps is  $\mbox{poly}(n,m,\frac{1}{\eps})$.
\end{proof}

\section{Improved Approximation for Additive Utilities}
\label{sec:additive}
In this section, we present a $9.27$-core for the special case of additive utilities. Recall that for such utilities, $u_i(S) = \sum_{j \in S} u_{ij}$, where $S \subseteq C$ is a subset of candidates and $i \in W$ is a voter. Though we could use the approach in the previous section, we lose constant factors first because the local optimum to the Nash Welfare objective only finds an approximate fractional core, and secondly because the randomized rounding needs to scale down budgets to satisfy the size constraints. 

We address the first issue by using an exact core solution to the fractional problem. We do this via the classic Lindahl equilibrium that we describe in \cref{sec:lindahl}. To address the second issue, we use dependent rounding that preserves the budget constraint with probability one for additive utilities, and we describe this in \cref{sec:satisfied}. This yields a $9.27$-core, though we do not know how to implement the resulting algorithm in polynomial time, since the fractional solution is now via a fixed point argument.

\subsection{Fractional Solution: Lindahl Equilibrium}
\label{sec:lindahl}
Our algorithm for constructing a committee in the approximate core will make use of the Lindahl equilibrium~\citep{lindahl1958just,foley1970lindahl}. This equilibrium yields a fractional committee that lies in the fractional core (\cref{def:core2} for $\alpha = 1$). 

We follow the approach in~\citep{DBLP:conf/wine/FainGM16} for specifying the Lindahl equilibrium. Let $x_j \ge 0$ denote the fraction to which candidate $j$ is chosen. Here, we will assume for technical reasons that this can be a quantity greater than $1$. We assign endowment $\frac{b}{n}$ to each voter $i \in V$, and a price $p_{ij}$ of $j \in C$ for $i \in V$. The Lindahl equilibrium is now defined as follows.

\begin{definition}[Lindahl Equilibrium]
\label{def:LindahlEq}
Let $p_{ij}$ be the price of $j \in C$ for $i \in V$, and let $x_j \ge 0$ be the fraction with which item $j$ is allocated. The prices and allocations constitute a \emph{Lindahl Equilibrium} if:
\begin{enumerate}
    \item For all $i \in V$, suppose the voter computes allocation $\mathbf{y} \ge 0$ that maximizes her utility $U_i(\mathbf{y})$ subject to her endowment constraint $\sum_{j \in C}p_{ij}y_{j} \leq \frac{b}{n}$, then we have that $\mathbf{y} = \mathbf{x}$.\label{item:lindhal_voter}
    \item The profit of the allocation given by $\sum_{i \in V}\sum_{j \in C} y_{j} p_{ij} - \sum_{j \in C} s_{j}y_{j}$, is maximized when $\mathbf{y} = \mathbf{x}$.\label{item:lindhal_profit}
\end{enumerate}
\end{definition}

The main result in~\citep{foley1970lindahl} is the following theorem proved by a fixed point argument. The setting in~\citep{foley1970lindahl} is much more general. Therefore, for the purpose of completeness, we present a simple and direct proof of this theorem and its consequences via complementarity theory~\citep{Eaves} in \cref{app:lindahlProof}.

\begin{restatable}[\citep{foley1970lindahl}]{theorem}{lindahl}
\label{thm:lindahlEq}
If the utility functions $U_i(\mathbf{x})$ have continuous derivatives, are strictly increasing, and are strictly concave, then a Lindahl equilibrium always exists. 
\end{restatable}

\subsubsection{Properties of the Lindahl Equilibrium}
We will need the following corollary to~\cref{thm:lindahlEq}, which for the purpose of completeness, we also prove in \cref{app:lindahlProof}.

\begin{restatable}{corollary}{lindahlcoro}
\label{coro:LindahlEqCoro}
It holds for a Lindahl equilibrium that:
 \begin{enumerate}
     \item For all voters $i \in V$, we have $\sum_j p_{ij} x_j = \frac{b}{n}$. \label{item:exhaust}
     \item For each candidate with $x_j > 0$, we have $\sum_i p_{ij} = s_j$, and for each candidate with $x_j = 0$, we have $\sum_i p_{ij} \le s_j$. \label{item:complementary}
     \item $\sum_{j \in C} s_j x_j = b$.
 \end{enumerate}   
\end{restatable}

Given \cref{coro:LindahlEqCoro}, it is easy and instructive to see that a Lindahl equilibrium lies in the core. Since this proof idea will be crucial to our analysis, we present it for completeness. 

\begin{corollary}[\citet{foley1970lindahl}]
\label{lem:lcore}
 The Lindahl equilibrium lies in the fractional core (\cref{def:core2} for $\alpha = 1$). 
\end{corollary}
\begin{proof}
Given the equilibrium $\mathbf{x}$, suppose for the purpose of contradiction that there exists a subset $S \subseteq V$ of $t$ voters that can deviate and choose a committee $\mathbf{z}$ of size $\sum_j s_j z_j \le \frac{t}{n} b$ such that $U_i(\mathbf{z}) > U_i(\mathbf{x})$ for all $i \in V$. Since $\mathbf{x}$ is utility maximizing for the voter $i$ at endowment $\frac{b}{n}$, and since by \cref{item:exhaust} of \cref{coro:LindahlEqCoro}, this endowment is spent exactly, the committee $\mathbf{z}$ must cost more than $b/n$. Therefore, for all $i \in V$, we have:
$$ \sum_j p_{ij} z_j > \frac{b}{n}.$$
However, by \cref{item:complementary} of \cref{coro:LindahlEqCoro}, $\sum_i p_{ij} \le s_j$ for all $i \in S, j \in C$. Summing the previous inequality over all $i \in S$, and applying $\sum_i p_{ij} \le s_j$, we obtain:
$$ \sum_j s_j z_j \ge \sum_{j, i \in V} p_{ij} z_j \ge \sum_{i \in S} \left( \sum_j p_{ij} z_j \right) > \sum_{i \in S} \frac{b}{n} = t \frac{b}{n}.$$
This contradicts the fact that $\mathbf{z}$ could have been purchased with the endowments of $i \in S$, that is, $\sum_j s_j z_j \le t \frac{b}{n}$, and thus completes the proof.
\end{proof}

\subsubsection{Lindahl Equilibrium for Additive Utilities}
So far, we have presented the Lindahl equilibrium in its generality for continuous, concave, non-decreasing utilities. We now specialize it to additive utilities. Given the additive utility function $u_i$, we make it continuous and unconstrained using the following natural definition:
\[
U_i(\mathbf{x}) := \sum_{j \in C} \min(1, x_j) \cdot u_i(\{j\}).
\]
It is clearly a concave function.\footnote{As a technicality, to make the function satisfy the preconditions in \cref{thm:lindahlEq}, we perturb $U_i$ slightly to make it be strictly increasing, strictly concave, and have continuous derivative. This perturbation preserves $U_i$ to within a $(1+\epsilon)$ approximation for $\epsilon > 0$ being arbitrary small, which will suffice for our approximation guarantees.} Given this utility function, the Lindahl equilibrium is defined exactly as in \cref{def:LindahlEq}, and this implies \cref{coro:LindahlEqCoro} holds as is.  The Lindahl equilibrium computation will be encapsulated by a subroutine below.

\begin{definition} [Subroutine \lin{}$(C, W, \{U_{i}\}, B)$]
Given the set of candidates $C$, a subset $W \subseteq V$ of voters with continuous and concave utilities $\{U_{i}\}$, and a size constraint $B$, this procedure finds a fractional solution $\mathbf{x} \ge 0$ over $C$ with $\sum_{j \in C} s_j x_j \le B$, such that this solution is a Lindahl equilibrium.
\end{definition} 

Note that the size $B$ could be different from $b$. Using \cref{def:core2}, the solution $\mathbf{x}$ of this subroutine satisfies the following condition: There is no subset $S \subseteq W$ of voters who can find a committee $\mathbf{z} \ge 0$ over $E$, with $\sum_{j \in E} s_j z_j \le \frac{|S|}{|W|} \cdot B$ such that $U_i(\mathbf{z}) > U_i(\mathbf{x})$ for all $i \in S$.

\subsection{Randomized Rounding and Satisfied Voters}
\label{sec:satisfied}
Note that if in the Lindahl equilibrium, we have $x_j > 1$ for some $j \in C$, then all items must be integrally allocated. This is because we assumed $U_i(\mathbf{x}) = u^+_i(\mathbf{y})$ where $y_j = \min(1,x_j)$. This means that if $x_j > 1$, then it can be reduced to $1$ without affecting any utilities. Therefore, if some item is fractionally allocated, we can decrease the allocation of $j$ with $x_j > 1$ and increase the allocation of a fractional item, improving the utility of some agent. This violates the core condition. We will therefore assume throughout that $x_j \in [0,1]$ for all $j \in C$.

\subsubsection{Rounding Procedure \idep{}$(C,W,\{u_i\},B)$}
Assume we are given a fractional solution $\mathbf{x} \ge 0$ to \lin{}$(C, W, \{u_{ij}\}_{i, j}, B)$. The procedure \idep{}$(C,W,\{u_{ij}\},B)$ uses the algorithm of~\citet{DBLP:journals/talg/ByrkaPRST17} to the fractionally allocated items in this solution. This produces a committee $(O,\ell)$, where the candidates in $O$ are chosen integrally, at most candidate $\ell$ is chosen fractionally, and the total size is always at most $B$.

Let $X_j$ be the random variable denoting whether candidate $j$ is selected: $X_j = 1$ if $j \in O$, $X_j = 0$ of $j \notin O \cup \{\ell\}$, and $X_j$ is the fractional weight if $j = \ell$. Then the procedure \idep{} preserves the marginals: $\Ex{}{X_j} = \min(1,x_j)$ for all $j \in C$, and the $\{X_j\}$ are negatively correlated~\citep{Gandhi,PanconesiS97}. Note that the candidates with $X_j = 1$ define $O$, while the one with $X_{\ell} \in (0,1)$ is the fractionally chosen candidate (it may not exist). 

Note that 
\begin{equation}
    \label{eq:expectedValueBound}
    \Ex{}{u_{i \ell} + \sum_{j \in O} u_{ij} } \ge \Ex{}{\sum_{j \in C} u_{ij} X_j} =  \sum_{j \in C} u_{ij} \min(1,x_j) = U_i(\mathbf{x})
\end{equation}
where the expectation is over the random choice of $O$. Further, the size constraint is never violated, so that: 
\[
\sum_{j \in O} s_j \le \sum_{j \in C} s_j X_j \le B
\]
regardless of the outcome of the rounding procedure.  Note finally that any candidate that is fully allocated by $\mathbf{x}$, that is, with $x_j \ge 1$, must be present in $O$.

\subsubsection{Satisfied Voters}
The procedure \idep{}$(C,W,\{u_{ij}\}, B)$ returns the committee $O$. Given the definition of $\gamma$-satisfied from \cref{def:gammaSatisfied}, we can strengthen \cref{thm:nw_satisfied} as follows:

\begin{theorem}[Constant Fraction of Constant-Satisfied Voters]
\label{thm:constant_satisfied2}
    Given the fractional solution $\mathbf{x}$ produced by \lin{}$(C, W, \{u_{ij}\}, B)$ where $|W| = n'$, there is a integral committee $O$ produced by \idep{}$(C, W, \{u_{ij}\}, B)$ with at least $(1 - \beta)n'$ $\gamma$-satisfied voters, where $\beta = \gamma e^{1-\gamma}$.
\end{theorem}

Proving \cref{thm:constant_satisfied2} requires a concentration bound for the sum of negatively correlated weighted Bernoulli random variables. 
\begin{restatable}[\citep{PanconesiS97}]{lemma}{chernoff}
\label{lem:chernoff}
Let $B_1, B_2, \ldots, B_k$ be $k$ negatively correlated Bernoulli random variables. Let $X = \sum_{i = 1}^k \beta_i B_i$, where $\beta_i \in [0,1]$. Let $\Ex{}{X} \ge \mu$.
Then for any constant $\delta \in (0, 1)$ we have:
\[
\Prob{}{X < (1 - \delta)\mu} < \left(\frac{e^{-\delta}}{(1 - \delta)^{(1 - \delta)}}\right)^{\mu}.
\]
\end{restatable}

Using the above bound, we complete the proof of \cref{thm:constant_satisfied2}.

\begin{proof}[Proof of \cref{thm:constant_satisfied2}]
Let $I \subseteq C$ be the set of fully allocated candidates in \lin{}$(C, W, \{u_{ij}\}_{i, j}, B)$ (i.e. $I = \{j \in C \mid x_j \geq 1\}$), and let $F \subseteq C$ be the set of fractionally allocated ones (i.e. $F = \{j \in C \mid x_j \in (0, 1)\}$). Let $u_i(F) = \sum_{j \in F} x_j u_{ij}$. There are two cases:

\noindent\textbf{Case (1):} Suppose these is some candidate $q \in F$ so that $u_{iq} \geq \frac{u_i(F)}{\gamma}$. In this case, voter $i$ is $\gamma$-satisfied with probability $1$, since
\[
u_{iq} + \sum_{j \in O} u_{ij} \geq \frac{\sum_{j \in F} x_ju_{ij}}{\gamma} + \sum_{j \in I} u_{ij} \geq \frac{\sum_{j \in E} u_{ij} \min(1, x_j)}{\gamma},
\]
where the first step uses that $I \subseteq O$.

\noindent\textbf{Case (2):} Suppose for every candidate $j \in F$, $u_{ij} < \frac{u_i(F)}{\gamma}$. In this case, we consider the candidate $\ell$ from the procedure \idep{} as the additament, and invoke our concentration bound of \cref{lem:chernoff}. Notice $X_j$'s are negatively correlated Bernoulli random variables. Let $\mu = \Ex{}{\sum_{j \in F} u_{ij} \cdot \frac{\gamma}{u_i(F)} \cdot X_j} = \gamma$ and $\delta = 1 - \frac{1}{\gamma}$. By \cref{lem:chernoff}, we have
\[
\Prob{}{\sum_{j \in F} u_{ij} \cdot \frac{\gamma}{u_i(F)} \cdot X_j < 1} < \left(\frac{e^{-\delta}}{(1 - \delta)^{(1 - \delta)}}\right)^{\mu} = \gamma e^{1 - \gamma}.
\]
If the event $\sum_{j \in F} u_{ij} \cdot \frac{\gamma}{u_i(F)} \cdot X_j < 1$ does not happen, the utility of $i$ with additament $\ell$ is at least
\begin{align*}
u_{i\ell} + \sum_{j \in O} u_{ij} &\geq u_{i\ell} X_\ell + \sum_{j \in F \setminus \{\ell\}} u_{ij} X_j + \sum_{j \in I} u_{ij} = \sum_{j \in F} u_{ij} X_j + \sum_{j \in I} u_{ij}\\
&\geq \frac{u_i(F)}{\gamma} + \sum_{j \in I} u_{ij} \geq \frac{\sum_{j \in E} u_{ij} \min(1, x_j)}{\gamma},
\end{align*}
where the first step again uses that $I \subseteq O$.
\end{proof}

\subsection{The Constant Approximation to the Core}
\label{sec:proof2}
In this section, we modify the algorithm in \cref{sec:proof} and prove that it returns a committee in an $O(1)$-core. The analysis of our algorithm now critically requires the market-clearing properties of the Lindahl equilibrium presented in \cref{def:LindahlEq}, \cref{coro:LindahlEqCoro}, and \cref{lem:lcore}.

\paragraph{Algorithm.} The algorithm is the same as \cref{alg:1}, except the following lines:

\begin{itemize}
    \item {\bf Line 7:} $\mathbf{x}_t \leftarrow \ $ \lin{}($C, V_t, \{U_{i}\}, b_t$).
    \item {\bf Line 8:} $O_t \leftarrow $ Solution of $\idep{}(C, V_t, \{u_i\}, b_t)$ that satisfies \cref{thm:constant_satisfied2}.
\end{itemize}

\paragraph{Analysis.} The analysis follows the same outline as that in \cref{sec:mainanalysis}. First, using the same argument as in that section, the solution $T^*$ is feasible for the size $b$. 

As before, we proceed to show a contradiction. Let $S$ denote the set of voters that deviate, and let $A$ denote the set of items they deviate to. We have $ \sum_{j \in A} s_j \le \frac{|S|}{n} \cdot b$, and $A$ provides an $\alpha$-factor larger utility to voters in $S$ even after including any additament. Consider the voters in $W_t$, and let $S_t = S \cap W_t$. These voters are $\gamma$-satisfied by $O_t$ with respect to the fractional solution $\mathbf{x}_t$ (\cref{def:gammaSatisfied}). Therefore, if $i \in S_t$ deviates to $A$ to obtain an $\alpha$-factor larger utility, it must be that $U_i(A) \ge \frac{\alpha}{\gamma} \cdot U_i(\mathbf{x}_t)$. Let $\theta = \frac{\alpha}{\gamma}$. We will assume $\theta > 1$ below.

We will now show the analog of \cref{thm:grad} using the prices computed by the Lindahl equilibrium. Let $p^t$ denote the prices computed by \lin{}($C, V_t, \{U_{i}\}, b_t)$, and let $n_t = |V_t|$.  The following lemma generalizes \cref{lem:lcore} and bounds the price of the set $A$ via the optimality conditions of the Lindahl equilibrium $\mathbf{x}_t$.

\begin{lemma}
\label{lem:priceA}
    For all $i \in S_t$, we have $\sum_{j \in A} p^t_{ij} \ge \theta \cdot \frac{b_t}{n_t}$.
\end{lemma}
\begin{proof}
    First note that $x_{tj} \le 1$ for all items $j$. If any $x_{tj} > 1$, then all items must be integrally allocated (see \cref{sec:lindahl}), which means $A$ cannot achieve a $\theta$-factor larger utility. Note that $\mathbf{x}$ is the utility maximizing solution to a packing problem for voter $i$ where the ``size'' of item $j$ is $p^t_{ij}$ and the ``size'' constraint is $\frac{b_t}{n_t}$. This constraint is exactly satisfied by \cref{item:exhaust} of \cref{coro:LindahlEqCoro}.  Since $U_i$ is concave, any solution that produces $\theta$ factor more utility must have ``size'' at least $\theta$ times larger. Therefore, the price of $A$ is $\theta$ times larger than that of $\mathbf{x}_t$, completing the proof.
\end{proof}

We now bound the size of set $S_t$ as follows:

\begin{lemma} [Analog of \cref{thm:grad}]
\label{lem:St}
$|S_t| \le \frac{n_t}{b_t} \cdot \frac{\sum_{j \in A} s_j}{\theta}$.
\end{lemma}
\begin{proof}
    Summing the bound in \cref{lem:priceA} over $i \in S_t$, we have
    \[
    \sum_{i \in S_t} \sum_{j \in A} p^t_{ij} \ge |S_t| \cdot \theta \cdot \frac{b_t}{n_t}.
    \]
    By \cref{item:complementary} of \cref{coro:LindahlEqCoro}, we have 
    \[
    \sum_{i \in V_t} p^t_{ij} \le s_j \ \ \forall j \qquad \implies \qquad \sum_{i \in S_t} \sum_{j \in A} p^t_{ij} \le \sum_{j \in A} s_j.
    \]
    Combining these inequalities completes the proof.
\end{proof}

This proves the analog of \cref{thm:grad}. Continuing as before, we sum \cref{lem:St} over all $t$, and using $ \sum_{j \in A} s_j \le \frac{|S|}{n} b$, we have:
\[
|S| = \sum_t |S_t| \le \frac{\sum_{j \in A} s_j}{\theta} \cdot \sum_t \frac{n_t}{b_t} \le \frac{|S|}{n} \cdot \frac{b}{\theta} \cdot \sum_t \frac{n_t}{b_t}.
\]
Therefore, for a blocking coalition to exist, we need:
\begin{equation} \label{eq:block11}
\frac{b}{n} \cdot \sum_t \frac{n_t}{b_t} \ge \theta = \frac{\alpha}{\gamma}.
\end{equation}
We will now set the parameters $\omega, \gamma, \alpha$ so that the above inequality is false. First note by \cref{thm:constant_satisfied2} that $n_{t+1} \le \beta n_t$ where $\beta = \gamma e^{1-\gamma}$. Further, $b_{t+1} = \omega b_t$. Therefore, 
\[
\frac{n_{t+1}}{b_{t+1}} \le \frac{\beta}{\omega} \cdot \frac{n_t}{b_t}
\]
with $\frac{n_0}{b_0} = \frac{n}{(1-\omega)b}$. Therefore,
\begin{equation} \label{eq:block21}
\frac{b}{n} \cdot \sum_t \frac{n_t}{b_t} \le \frac{1}{(1-\omega)} \cdot \sum_{t \ge 0} \left( \frac{\beta}{\omega} \right)^t = \frac{\omega}{(1-\omega)(\omega - \beta)}.
\end{equation}

Combining \cref{eq:block11} and \cref{eq:block21}, for a blocking coalition to exist, we need
\[
\alpha \le \frac{\omega \gamma}{(1-\omega)\left(\omega - \beta\right)} = \frac{\omega \gamma}{(1-\omega)\left(\omega - \gamma e^{1-\gamma}\right)},
\]
where the last step uses $\beta = \gamma e^{1-\gamma}$ according to \cref{thm:constant_satisfied2}.

For an $\alpha$ slightly larger than the right-hand side, a blocking coalition will therefore not exist. Plugging $\omega = 0.15$ and $\gamma = 6.7$ shows $\alpha < 9.27$, which yields the following theorem:
\additive*



\section{Lower Bounds}
\label{sec:general}
In this section, we provide lower bound examples for general monotone utilities and monotone submodular utilities. The former result rules out extending our constant factor bound to general monotone utilities, while the latter shows that for submodular utilities, there is a lower bound on approximation of an absolute constant $c > 1$.

\subsection{General Monotone Utilities}
\general*

\cref{thm:general_lb} is proved by the following example. (The same structure of $2$ groups of $3$ cyclically symmetric voters appears in \citep{DBLP:conf/ec/FainMS18, DBLP:journals/corr/PetersPS20}.)
\begin{example}
\label{ex:general}
We have $n = 6$ voters and $m = 30$ candidates. The candidates are grouped into $6$ disjoint sets, each of which contains $5$ candidates and is called a ``gadget''. We name the gadgets $g_1, \ldots, g_6$. Each voter $i$ has a favorite gadget $g_{f_i}$ and a second favorite gadget $g_{s_i}$, given by:
\begin{alignat*}{11}
&f_1 = 1, \ &&f_2 = 2, \ &&f_3 = 3, \ &&f_4 = 4, \ &&f_5 = 5, \ &&f_6 = 6;\\
&s_1 = 2, \ &&s_2 = 3, \ &&s_3 = 1, \ &&s_4 = 5, \ &&s_5 = 6, \ &&s_6 = 4.
\end{alignat*}
For any committee $E$, let $x_i(E) = \frac{1}{5}|E \cap g_{f_i}|$ and $y_i(E) = \frac{1}{5}|E \cap g_{s_i}|$, denoting the fraction of candidates in the favorite / second favorite gadget of voter $i$ being selected into $E$, respectively. The utility of voter $i$ on $E$ is given by
\[
u_i(E) = (\alpha + 1) \cdot \mathbbm{1}[x_i(E) = 1] + \mathbbm{1}[y_i(E) = 1].
\]
Here $\alpha = \varphi(n, m) = \varphi(6, 30)$. Her utility is monotone and supermodular.

Each candidate is of unit size $1$ and the budget $b = 15$. For any feasible committee $E$, there must be at least $3$ gadgets $g$'s with $|E \cap g| \leq 3$ -- otherwise the committee has at least $4 \cdot 4 = 16 > b$ candidates. Therefore, either $\{g_1, g_2, g_3\}$ or $\{g_4, g_5, g_6\}$ includes at least $2$ gadgets with $|E \cap g| \leq 3$. Without loss of generality, assume $\{g_1, g_2, g_3\}$ does and $|E \cap g_1| \leq 3$, $|E \cap g_2| \leq 3$. In this case, voters $1$ and $2$ can deviate and buy $g_2$, as they have a budget of $b \cdot \frac{2}{n} = 5$. For any additaments $q$ and $q'$:
\begin{alignat*}{3}
&u_1(g_2) = 1, \quad &&u_1(E \cup \{q\}) = 0;\\
&u_2(g_2) = \alpha + 1, \quad &&u_2(E \cup \{q'\}) \leq 1.
\end{alignat*}
We have $u_1(g_2) > \alpha u_1(E \cup \{q\})$ and $u_2(g_2) > \alpha u_2(E \cup \{q'\})$.
\end{example}

\subsection{Submodular Utilities}
Next, we modify \cref{ex:general} to show a lower bound for monotone submodular utilities.
\submodularlb*
\begin{example}
\label{ex:submodular}
We use the same setting as \cref{ex:general}, except the utility of each voter $i$ is given by
\[
u_i(E) = x_i(E) + z \cdot (1 - x_i(E)) \cdot y_i(E),
\]
where $z \in (0, 1)$ is a constant to be determined later.

\begin{lemma}
The function $u_i$ is monotone and submodular.
\end{lemma} 
\begin{proof}
Fix any $E \subseteq C$ and $t \notin E$ and consider $u_i(E \cup \{t\}) - u_i(T)$. Since $g_{f_i} \cap g_{s_i} = \varnothing$, this $t$ lies in one of these two sets but not both. Suppose $t \in g_{s_i}$. Then, \[
u_i(E \cup \{t\}) - u_i(E) = z \cdot (1 - x_i(E)) \cdot \left( y_i(E \cup \{t\}) - y_i(E) \right) \ge 0.
\]
Similarly, if $t \in g_{f_i}$, we have
$$ u_i(E \cup \{t\}) - u_i(E) = (1 - z y_i(E)) \cdot \left( x_i(E \cup \{t\}) - x_i(E) \right) \ge 0,$$
where we have used that $y_i(E) \le 1$ and $z \in [0,1]$. This shows that $u_i$ is a monotone function.

Similarly, if $E \subseteq E'$, then  $1 - x_i(E) \ge 1 - x_i(E')$ since the coverage function is monotone. Further, by the submodularity of the coverage function, we have
\[
y_i(E \cup \{t\}) - y_i(E) \ge y_i(E' \cup \{t\}) - y_i(E').
\]
Therefore, if $t \in g_{s_i}$, then
\begin{align*} 
u_i(E \cup \{t\}) - u_i(E) & = z \cdot (1 - x_i(E)) \cdot \left( y_i(E \cup \{t\}) - y_i(E) \right) \\
  & \ge z \cdot (1 - x_i(E')) \cdot \left( y_i(E' \cup \{t\}) - y_i(E') \right) \\
  & = u_i(E' \cup \{t\}) - u_i(E').
\end{align*}
A similar argument for the case where $t \in g_{f_i}$ completes the proof of submodularity.
\end{proof}

For any feasible committee $E$, again without loss of generality, assume $|E \cap g_1| \leq 3$, $|E \cap g_2| \leq 3$. In this case, voters $1$ and $2$ can deviate and buy $g_2$. For any additaments $q$ and $q'$:
\begin{alignat*}{3}
&u_1(g_2) = z, \quad  &&u_1(E \cup \{q\}) \leq 0.8 + z \cdot 0.2 \cdot 0.6;\\
&u_2(g_2) = 1, \quad  &&u_2(E \cup \{q'\}) \leq 0.8 + z \cdot 0.2 \cdot 1.
\end{alignat*}
When $z = \frac{\sqrt{689} - 17}{10} \approx 0.925$, the gap is $\min\left(\frac{u_1(g_2)}{u_1(E \cup \{q\})}, \frac{u_2(g_2)}{u_2(E \cup \{q'\})}\right) \geq \frac{5\sqrt{689} - 115}{16} > 1.015$.
\end{example}

\section{Conclusions and Open Questions}
Our work brings up several open questions. First is the existence of polynomial time computable pricing rules that approximate the core. The work of~\citet*{DBLP:journals/corr/PetersPS20} shows a price increase process that runs in polynomial time and provides a logarithmic approximation to the core for additive utilities. However, the prices are {\em common} to the voters, as opposed to the per-voter market clearing prices of the Lindahl equilibrium, and we do not know how to compute the latter in polynomial time. Is there an intermediate price tattonnement scheme that not only runs in polynomial time, but also achieves a constant approximation?

Next, our approximation bounds are far from tight. In particular, we have not ruled out the existence of a $1$-core for additive utility, or an $O(1)$-core for subadditive or XOS utilities. To address the latter, we would need to understand the approximability of the core when utilities of voters are continuous and the maximum of linear functions. Such utilities are convex, but have a specific form that may be amenable to better approximation results than the worst case illustrated in \cref{thm:general_lb}. We leave this as an interesting open question.

\bibliographystyle{plainnat}
\bibliography{citations}

\begin{thebibliography}{44}
\providecommand{\natexlab}[1]{#1}
\providecommand{\url}[1]{\texttt{#1}}
\expandafter\ifx\csname urlstyle\endcsname\relax
  \providecommand{\doi}[1]{doi: #1}\else
  \providecommand{\doi}{doi: \begingroup \urlstyle{rm}\Url}\fi

\bibitem[pbs()]{pbstanford}
{The Stanford Participatory Budgeting Platform}.
\newblock \url{https://pbstanford.org}.

\bibitem[Agrawal et~al.(2010)Agrawal, Ding, Saberi, and Ye]{Shipra10}
Shipra Agrawal, Yichuan Ding, Amin Saberi, and Yinyu Ye.
\newblock Correlation robust stochastic optimization.
\newblock In \emph{SODA}, page 1087–1096, 2010.

\bibitem[Arrow and Debreu(1954)]{ArrowD}
Kenneth~J. Arrow and Gerard Debreu.
\newblock Existence of an equilibrium for a competitive economy.
\newblock \emph{Econometrica}, 22\penalty0 (3):\penalty0 265--290, 1954.

\bibitem[Aziz and Shah(2021)]{aziz2021participatory}
Haris Aziz and Nisarg Shah.
\newblock Participatory budgeting: models and approaches.
\newblock In \emph{Pathways Between Social Science and Computational Social
  Science}, pages 215--236. 2021.

\bibitem[Aziz et~al.(2017)Aziz, Brill, Conitzer, Elkind, Freeman, and
  Walsh]{JR}
Haris Aziz, Markus Brill, Vincent Conitzer, Edith Elkind, Rupert Freeman, and
  Toby Walsh.
\newblock Justified representation in approval-based committee voting.
\newblock \emph{Social Choice and Welfare}, 48\penalty0 (2):\penalty0 461--485,
  2017.

\bibitem[Aziz et~al.(2018)Aziz, Elkind, Huang, Lackner, Fern{\'{a}}ndez, and
  Skowron]{PJR2018}
Haris Aziz, Edith Elkind, Shenwei Huang, Martin Lackner, Luis~S{\'{a}}nchez
  Fern{\'{a}}ndez, and Piotr Skowron.
\newblock On the complexity of extended and proportional justified
  representation.
\newblock In \emph{AAAI}, pages 902--909, 2018.

\bibitem[Aziz et~al.(2019)Aziz, Brandt, Elkind, and Skowron]{AzizChapter}
Haris Aziz, Felix Brandt, Edith Elkind, and Piotr Skowron.
\newblock \emph{Computational Social Choice: The First Ten Years and Beyond},
  pages 48--65.
\newblock Springer International Publishing, 2019.

\bibitem[Barman et~al.(2018)Barman, Krishnamurthy, and Vaish]{BarmanK}
Siddharth Barman, Sanath~Kumar Krishnamurthy, and Rohit Vaish.
\newblock Finding fair and efficient allocations.
\newblock In \emph{EC}, page 557–574, 2018.

\bibitem[Brainard and Scarf(2005)]{Fisher}
William~C. Brainard and Herbert~E. Scarf.
\newblock How to compute equilibrium prices in 1891.
\newblock \emph{American Journal of Economics and Sociology}, 64\penalty0
  (1):\penalty0 57--83, 2005.

\bibitem[Brams et~al.(2007)Brams, Kilgour, and Sanver]{Brams2007}
Steven~J. Brams, D.~Marc Kilgour, and M.~Remzi Sanver.
\newblock A minimax procedure for electing committees.
\newblock \emph{Public Choice}, 132\penalty0 (3):\penalty0 401--420, 2007.

\bibitem[Brandt et~al.(2016)Brandt, Conitzer, Endriss, Lang, and
  Procaccia]{VinceBook}
Felix Brandt, Vincent Conitzer, Ulle Endriss, J\'{e}r\^{o}me Lang, and Ariel~D.
  Procaccia.
\newblock \emph{Handbook of Computational Social Choice}.
\newblock Cambridge University Press, USA, 1st edition, 2016.
\newblock ISBN 1107060435.

\bibitem[Byrka et~al.(2017)Byrka, Pensyl, Rybicki, Srinivasan, and
  Trinh]{DBLP:journals/talg/ByrkaPRST17}
Jaroslaw Byrka, Thomas~W. Pensyl, Bartosz Rybicki, Aravind Srinivasan, and Khoa
  Trinh.
\newblock An improved approximation for $k$-median and positive correlation in
  budgeted optimization.
\newblock \emph{{ACM} Trans. Algorithms}, 13\penalty0 (2):\penalty0
  23:1--23:31, 2017.

\bibitem[Cabannes(2004)]{cabannes2004participatory}
Yves Cabannes.
\newblock Participatory budgeting: a significant contribution to participatory
  democracy.
\newblock \emph{Environment and urbanization}, 16\penalty0 (1):\penalty0
  27--46, 2004.

\bibitem[C{\u{a}}linescu et~al.(2011)C{\u{a}}linescu, Chekuri, P{\'{a}}l, and
  Vondr{\'{a}}k]{CalinescuCPV11}
Gruia C{\u{a}}linescu, Chandra Chekuri, Martin P{\'{a}}l, and Jan
  Vondr{\'{a}}k.
\newblock Maximizing a monotone submodular function subject to a matroid
  constraint.
\newblock \emph{{SIAM} J. Comput.}, 40\penalty0 (6):\penalty0 1740--1766, 2011.

\bibitem[Chamberlin and Courant(1983)]{CC}
John~R. Chamberlin and Paul~N. Courant.
\newblock Representative deliberations and representative decisions:
  Proportional representation and the borda rule.
\newblock \emph{The American Political Science Review}, 77\penalty0
  (3):\penalty0 718--733, 1983.

\bibitem[Chekuri et~al.(2010)Chekuri, Vondr{\'{a}}k, and
  Zenklusen]{ChekuriVZ10}
Chandra Chekuri, Jan Vondr{\'{a}}k, and Rico Zenklusen.
\newblock Dependent randomized rounding via exchange properties of
  combinatorial structures.
\newblock In \emph{FOCS}, pages 575--584, 2010.

\bibitem[Chekuri et~al.(2014)Chekuri, Vondr{\'{a}}k, and
  Zenklusen]{ChekuriVZ14}
Chandra Chekuri, Jan Vondr{\'{a}}k, and Rico Zenklusen.
\newblock Submodular function maximization via the multilinear relaxation and
  contention resolution schemes.
\newblock \emph{{SIAM} J. Comput.}, 43\penalty0 (6):\penalty0 1831--1879, 2014.

\bibitem[Chen et~al.(2019)Chen, Fain, Lyu, and Munagala]{ChenFLM19}
Xingyu Chen, Brandon Fain, Liang Lyu, and Kamesh Munagala.
\newblock Proportionally fair clustering.
\newblock In \emph{ICML}, pages 1032--1041, 2019.

\bibitem[Cheng et~al.(2020)Cheng, Jiang, Munagala, and
  Wang]{DBLP:journals/teac/ChengJMW20}
Yu~Cheng, Zhihao Jiang, Kamesh Munagala, and Kangning Wang.
\newblock Group fairness in committee selection.
\newblock \emph{{ACM} Trans. Economics and Comput.}, 8\penalty0 (4):\penalty0
  23:1--23:18, 2020.

\bibitem[Cole and Gkatzelis(2015)]{ColeG}
Richard Cole and Vasilis Gkatzelis.
\newblock Approximating the nash social welfare with indivisible items.
\newblock \emph{SIGecom Exch.}, 14\penalty0 (1):\penalty0 84–88, November
  2015.

\bibitem[Droop(1881)]{Droop}
H.~R. Droop.
\newblock On methods of electing representatives.
\newblock \emph{Journal of the Statistical Society of London}, 44\penalty0
  (2):\penalty0 141--202, 1881.

\bibitem[Eaves(1971)]{Eaves}
B.~C. Eaves.
\newblock On the basic theorem of complementarity.
\newblock \emph{Mathematical Programming}, 1\penalty0 (1):\penalty0 68--75,
  1971.

\bibitem[Eisenberg and Gale(1959)]{EG}
Edmund Eisenberg and David Gale.
\newblock Consensus of subjective probabilities: The pari-mutuel method.
\newblock \emph{The Annals of Mathematical Statistics}, 30\penalty0
  (1):\penalty0 165--168, 1959.

\bibitem[Endriss(2017)]{EndrissBook}
Ulle Endriss.
\newblock \emph{Trends in Computational Social Choice}.
\newblock Lulu.com, 2017.
\newblock ISBN 1326912097.

\bibitem[Fain et~al.(2016)Fain, Goel, and Munagala]{DBLP:conf/wine/FainGM16}
Brandon Fain, Ashish Goel, and Kamesh Munagala.
\newblock The core of the participatory budgeting problem.
\newblock In \emph{WINE}, pages 384--399, 2016.

\bibitem[Fain et~al.(2018)Fain, Munagala, and Shah]{DBLP:conf/ec/FainMS18}
Brandon Fain, Kamesh Munagala, and Nisarg Shah.
\newblock Fair allocation of indivisible public goods.
\newblock In \emph{EC}, pages 575--592, 2018.

\bibitem[Fern{\'{a}}ndez et~al.(2017)Fern{\'{a}}ndez, Elkind, Lackner,
  Garc{\'{\i}}a, Arias{-}Fisteus, Basanta{-}Val, and Skowron]{Sanchez}
Luis~S{\'{a}}nchez Fern{\'{a}}ndez, Edith Elkind, Martin Lackner,
  Norberto~Fern{\'{a}}ndez Garc{\'{\i}}a, Jes{\'{u}}s Arias{-}Fisteus, Pablo
  Basanta{-}Val, and Piotr Skowron.
\newblock Proportional justified representation.
\newblock In \emph{AAAI}, pages 670--676, 2017.

\bibitem[Foley(1970)]{foley1970lindahl}
Duncan~K Foley.
\newblock Lindahl's solution and the core of an economy with public goods.
\newblock \emph{Econometrica}, pages 66--72, 1970.

\bibitem[Friedman et~al.(2019)Friedman, Gkatzelis, Psomas, and Shenker]{Psomas}
Eric~J. Friedman, Vasilis Gkatzelis, Christos{-}Alexandros Psomas, and Scott
  Shenker.
\newblock Fair and efficient memory sharing: Confronting free riders.
\newblock In \emph{AAAI}, pages 1965--1972, 2019.

\bibitem[Gandhi et~al.(2006)Gandhi, Khuller, Parthasarathy, and
  Srinivasan]{Gandhi}
Rajiv Gandhi, Samir Khuller, Srinivasan Parthasarathy, and Aravind Srinivasan.
\newblock Dependent rounding and its applications to approximation algorithms.
\newblock \emph{J. ACM}, 53\penalty0 (3):\penalty0 324–360, 2006.

\bibitem[Goel et~al.(2019)Goel, Krishnaswamy, Sakshuwong, and
  Aitamurto]{knapsackVoting}
Ashish Goel, Anilesh~K. Krishnaswamy, Sukolsak Sakshuwong, and Tanja Aitamurto.
\newblock Knapsack voting for participatory budgeting.
\newblock \emph{ACM Trans. Econ. Comput.}, 7\penalty0 (2), July 2019.

\bibitem[Jiang et~al.(2020)Jiang, Munagala, and Wang]{DBLP:conf/stoc/JiangMW20}
Zhihao Jiang, Kamesh Munagala, and Kangning Wang.
\newblock Approximately stable committee selection.
\newblock In \emph{STOC}, pages 463--472, 2020.

\bibitem[Kunjir et~al.(2017)Kunjir, Fain, Munagala, and Babu]{ROBUS}
Mayuresh Kunjir, Brandon Fain, Kamesh Munagala, and Shivnath Babu.
\newblock {ROBUS:} fair cache allocation for data-parallel workloads.
\newblock In \emph{SIGMOD}, pages 219--234, 2017.

\bibitem[Lindahl(1958)]{lindahl1958just}
Erik Lindahl.
\newblock Just taxation—a positive solution.
\newblock In \emph{Classics in the theory of public finance}, pages 168--176.
  1958.

\bibitem[Monroe(1995)]{Monroe}
Burt~L. Monroe.
\newblock Fully proportional representation.
\newblock \emph{The American Political Science Review}, 89\penalty0
  (4):\penalty0 925--940, 1995.

\bibitem[Munagala et~al.(2021)Munagala, Shen, and
  Wang]{DBLP:conf/ec/MunagalaSW21}
Kamesh Munagala, Zeyu Shen, and Kangning Wang.
\newblock Optimal algorithms for multiwinner elections and the
  chamberlin-courant rule.
\newblock In \emph{EC}, pages 697--717, 2021.

\bibitem[Nash(1950)]{Nash}
John~F. Nash.
\newblock The bargaining problem.
\newblock \emph{Econometrica}, 18\penalty0 (2):\penalty0 155--162, 1950.

\bibitem[Panconesi and Srinivasan(1997)]{PanconesiS97}
Alessandro Panconesi and Aravind Srinivasan.
\newblock Randomized distributed edge coloring via an extension of the
  chernoff-hoeffding bounds.
\newblock \emph{{SIAM} J. Comput.}, 26\penalty0 (2):\penalty0 350--368, 1997.

\bibitem[Peters and Skowron(2020)]{PetersS20}
Dominik Peters and Piotr Skowron.
\newblock Proportionality and the limits of welfarism.
\newblock In \emph{EC}, pages 793--794, 2020.

\bibitem[Peters et~al.(2021)Peters, Pierczy{\'n}ski, and
  Skowron]{DBLP:journals/corr/PetersPS20}
Dominik Peters, Grzegorz Pierczy{\'n}ski, and Piotr Skowron.
\newblock Proportional participatory budgeting with additive utilities.
\newblock In \emph{NeurIPS}, 2021.

\bibitem[Thiele(1895)]{thiele1895om}
Thorvald~N Thiele.
\newblock Om flerfoldsvalg.
\newblock \emph{Oversigt over det Kongelige Danske Videnskabernes Selskabs
  Forhandlinger}, 1895:\penalty0 415--441, 1895.

\bibitem[Tideman and Richardson(2000)]{Tideman}
Nicolaus Tideman and Daniel Richardson.
\newblock Better voting methods through technology: The
  refinement-manageability trade-off in the single transferable vote.
\newblock \emph{Public Choice}, 103\penalty0 (1):\penalty0 13--34, 2000.

\bibitem[Vondrak(2008)]{Vondrak}
Jan Vondrak.
\newblock Optimal approximation for the submodular welfare problem in the value
  oracle model.
\newblock In \emph{STOC}, page 67–74, 2008.

\bibitem[Yan(2011)]{Yan11}
Qiqi Yan.
\newblock Mechanism design via correlation gap.
\newblock In \emph{SODA}, page 710–719, 2011.

\end{thebibliography}

\appendix
\section{Existence of Lindahl Equilibrium}
\label{app:lindahlProof}
For the purpose of completeness, we present a direct proof of the existence of Lindahl equilibrium, showing \cref{thm:lindahlEq} and \cref{coro:LindahlEqCoro}. We restate \cref{thm:lindahlEq} below.

\lindahl*

Following~\citep*{DBLP:conf/wine/FainGM16}, define 
$$ f_j(\mathbf{x}) = - \sum_{i \in V} \frac{\frac{\partial U_i(\mathbf{x})}{\partial x_j}}{\sum_{\ell \in C} x_{\ell} \frac{\partial U_i(\mathbf{x})}{\partial x_{\ell}}}. $$

Applying the constrained complementarity theorem of~\cite{Eaves} to this continuous function $\{f_j\}$ with the constraint $\sum_{j \in C} s_j x_j \le b$, there exists $\mathbf{x}, \mathbf{y} \ge 0$, and scalar $z \ge 0$ such that:
\begin{itemize}
    \item For all $j \in C$, $ f_j + z s_j = y_j$; and $\mathbf{x} \cdot \mathbf{y} = 0$.
    \item $\sum_{j \in C} x_j s_j \le b$, and $z ( b - \sum_{j \in C} x_j s_j) = 0$. 
\end{itemize}

In the sequel, we will focus on this solution. We will first show that $z > 0$ so that $\sum_{j \in C} s_j x_j = b$. Otherwise, for all $j \in C$, we have $f_j = y_j$. However, $f_j < 0$ since the utilities are strictly increasing, while $y_j \ge 0$ by assumption. This is a contradiction.

Let $C' = \{j \in C \mid x_j > 0\}$. For $j \in C'$, we have $y_j = 0$ so that $f_j = - z s_j$. Multiplying by $x_j$ and summing, we have
$$ \sum_{j \in C'} x_j f_j = - \sum_{i \in V} \frac{\sum_{j \in C'} x_j \frac{\partial U_i(\mathbf{x})}{\partial x_j}}{\sum_{\ell \in C} x_{\ell} \frac{\partial U_i(\mathbf{x})}{\partial x_{\ell}}} = - \sum_{i \in V} 1 =  - n = - z \sum_{j \in C'} s_j x_j = - z b.$$
Therefore, $z = n/b$.  

Set $p_{ij} = \frac{b}{n} \frac{\frac{\partial U_i(\mathbf{x})}{\partial x_j}}{\sum_{\ell \in C} x_{\ell} \frac{\partial U_i(\mathbf{x})}{\partial x_{\ell}}}$ for all $i \in V, j \in C$. Note that $\sum_i p_{ij} = - \frac{b}{n} f_j$. This implies the following:
\begin{itemize}
    \item For all $i \in V$, we have $\sum_j p_{ij} x_j = \frac{b}{n} \frac{\sum_{j \in C} x_j \frac{\partial U_i(\mathbf{x})}{\partial x_j}}{\sum_{\ell \in C} x_{\ell} \frac{\partial U_i(\mathbf{x})}{\partial x_{\ell}}} =  \frac{b}{n}$.
    \item For all $j \in C'$, since $f_j + z s_j = 0$ and $z = \frac{b}{n}$, we have $\sum_i p_{ij} = s_j$. 
    \item For all $j \in C$, since $f_j + z s_j = y_j \ge 0$, we have $\sum_i p_{ij} \le s_j$.
    \item For all $i \in V$ and for all $j,\ell \in C$, we have $p_{ij} \frac{\partial U_i}{\partial x_{\ell}} = p_{i \ell} \frac{\partial U_i}{\partial x_{j}}$.
\end{itemize}

The first three consequences, along with $\sum_{j \in C} s_j x_j = b$ prove \cref{coro:LindahlEqCoro}. By simple gradient optimality, these conditions also show \cref{item:lindhal_profit} of \cref{def:LindahlEq}. To see this, note that \cref{def:LindahlEq} does not constrain the allocation $\mathbf{y}$. Therefore, for the profit to be finite, we have $\sum_i p_{ij} - s_j \le 0$ for all $j$. Further, if any of these inequalities is strict, the profit is only larger if $x_j = 0$. Therefore, if $x_j > 0$, it must force this inequality to be tight. These are exactly the conditions we derived above, which means this solution satisfies \cref{item:lindhal_profit}. Similarly, since $U_i$ is strictly increasing and concave, the last condition derived above is the gradient optimality condition for \cref{item:lindhal_voter} in \cref{def:LindahlEq}.  To see this, note that the gradient optimality condition of \cref{item:lindhal_voter} can be written as $\frac{\partial U_i}{\partial x_j} = \lambda_i p_{ij}$ for all $i \in V, j \in C$. Therefore, $p_{ij} \frac{\partial U_i}{\partial x_{\ell}} = p_{i \ell} \frac{\partial U_i}{\partial x_{j}}$, so that any solution satisfying the latter and with $\sum_j p_{ij} x_j = \frac{b}{n}$ must satisfy \cref{item:lindhal_voter}.

\end{document}